\documentclass[a4paper,UKenglish,cleveref,autoref,thm-restate]{lipics-v2021}

\title{Listing $6$-Cycles in Sparse Graphs}

\author{Virginia {Vassilevska Williams}}{MIT, USA}{virgi@mit.edu}{https://orcid.org/0000-0003-4844-2863}{Supported by NSF Grant CCF-2330048, BSF Grant 2020356 and a Simons Investigator Award.}

\author{Alek Westover}{MIT, USA}{alekw@mit.edu}{https://orcid.org/0009-0007-8381-5705}{}

\authorrunning{V. {Vassilevska Williams} and A. Westover} 

\Copyright{ Virginia {Vassilevska Williams} and Alek Westover } 



\category{} 

\relatedversion{} 

\supplement{}
\supplementdetails[subcategory={}, cite={}, swhid={}]{Software}{https://github.com/awestover/listing-C6s-LP}

\acknowledgements{We'd like to thank Nathan S. Sheffield, Claire Zhang, and Ce Jin for helpful discussions.}

\nolinenumbers 

\EventEditors{Raghu Meka}
\EventNoEds{1}
\EventLongTitle{16th Innovations in Theoretical Computer Science Conference (ITCS 2025)}
\EventShortTitle{ITCS 2025}
\EventAcronym{ITCS}
\EventYear{2025}
\EventDate{January 7--10, 2025}
\EventLocation{Columbia University, New York, NY, USA}
\EventLogo{}
\SeriesVolume{325}
\ArticleNo{6}

\usepackage{amsthm}
\usepackage{amssymb}
\usepackage{amsmath}

\newcommand{\defn}[1]{\emph{\textbf{{#1}}}}
\renewcommand{\paragraph}[1]{\vspace{.15 cm} \noindent \textbf{#1} }

\usepackage{bbm}
\newcommand{\les}{\lesssim}
\newcommand{\mc}{\mathcal}

\newcommand{\one}{\mathbbm{1}}
\renewcommand{\O}{O}
\newcommand{\tilo}{\widetilde{\O}}

\DeclareMathOperator{\polylog}{\text{polylog}}
\DeclareMathOperator{\poly}{\text{poly}}

\newcommand{\interior}[1]{ {\kern0pt#1}^{\mathrm{o}} }

\newcommand{\eps}{\varepsilon}

\newcommand{\setof}[2]{\left\{ #1\; \mid \;#2 \right\}}
\newcommand{\set}[1]{\left\{ #1\right\}}

\newcommand{\N}{\mathbb{N}}

\newcommand{\ceil}[1]{\lceil #1 \rceil}

\newcommand{\abs}[1]{\left| #1 \right|}

\theoremstyle{definition}
\newtheorem{fact}[theorem]{Fact}

\usepackage{graphicx}
\usepackage{subcaption}
\usepackage{listings}
\begin{document}

\maketitle

\begin{abstract}
This work considers the problem of output-sensitive listing of occurrences of
$2k$-cycles for fixed constant $k\geq 2$ in an undirected host graph with $m$ edges and $t$ $2k$-cycles.
Recent work of Jin and Xu (and independently Abboud, Khoury, Leibowitz, and Safier) [STOC 2023] gives an $O(m^{4/3}+t)$ time algorithm for listing $4$-cycles,
and recent work by Jin, Vassilevska Williams and Zhou [SOSA 2024] gives an
$\widetilde{O}(n^2+t)$ time algorithm for listing $6$-cycles in $n$ node graphs.
We focus on resolving the next natural question: obtaining listing algorithms
for $6$-cycles in the sparse setting, i.e., in terms of $m$ rather than $n$.
Previously, the best known result here is the better of Jin, Vassilevska Williams and Zhou's $\tilo(n^2+t)$ algorithm and Alon, Yuster and Zwick's $O(m^{5/3}+t)$ algorithm.

We give an algorithm for listing $6$-cycles with running time $\tilo(m^{1.6}+t)$.
Our algorithm is a natural extension of Dahlgaard, Knudsen and
St\"ockel's [STOC 2017] algorithm for detecting a $2k$-cycle.
Our main technical contribution is the analysis of the algorithm which involves
a type of ``supersaturation'' lemma relating the number of $2k$-cycles in a
bipartite graph to the sizes of the parts in the bipartition and the number of
edges.
We also give a simplified analysis of Dahlgaard, Knudsen and St\"ockel's
$2k$-cycle detection algorithm (with a small polylogarithmic increase in the
running time), which is helpful in analyzing our listing algorithm.
\end{abstract}



\section{Introduction}
Listing copies of a small pattern graph that occur as subgraphs of a host graph
is a fundamental problem in algorithmic graph theory. In this work, we consider
the problem of listing $C_{2k}$'s (i.e., $2k$-cycles) for fixed constant $k\geq 2$.
Some examples of applications of $C_{2k}$ listing include analyzing social
networks \cite{motivationSocial}, and understanding causal relationships in
biological interaction graphs \cite{motivationBIO}. (See e.g. \cite{dense_cycle_detection_YZ} for further motivation.)
In the following discussion we consider an $n$ vertex $m$ edge graph $G$ with
$t$ $C_{2k}$'s (where $k\geq 2$ will be clear from context and $t$ is not necessarily known).

The main reason for focusing on even length cycles is that while there are
dense graphs that contain no odd cycles (e.g. bipartite graphs), a classic result of Bondy and
Simonovits \cite{evenextremal} states that any graph with at least
$100kn^{1+1/k}$ edges contains a $C_{2k}$, for any integer $k\geq 2$. This fact enables efficient
``combinatorial'' algorithms for $C_{2k}$ detection, in contrast to the case of odd
length cycles where efficient $C_{2k+1}$ detection algorithms are based on
matrix multiplication. In fact, very simple reductions (e.g. \cite{vthesis}) show that for any $k\geq 1$, $C_{2k+1}$ detection is at least as hard as triangle detection, and the latter problem is known to be fine-grained subcubically equivalent to Boolean Matrix Multiplication \cite{focsy}. Thus, fast algorithms for odd cycles cannot avoid matrix multiplication. We focus on even cycles from now on.

A classic result of Yuster and Zwick \cite{dense_cycle_detection_YZ} shows how
to determine whether a graph contains a $C_{2k}$ in $O(n^2)$ time for any given constant $k\geq 2$. This quadratic running time is conjectured
to be optimal in general (see \cite{short_cycle_removal,LincolnVyas,Kn17}, also discussed below); in particular, \cite{short_cycle_removal} gave concrete evidence for the hardness of $C_4$ detection.

Nevertheless, in
the regime of sparser graphs, where $m<o(n^{1+1/k})$ improvements are possible. For $C_4$'s there
is a simple $O(m^{4/3})$ algorithm \cite{alon1997finding} that matches the
performance of the $O(n^2)$ algorithm
of \cite{dense_cycle_detection_YZ} for $m=\Theta(n^{3/2})$, and improves on
the performance for all $m<o(n^{3/2})$. Dahlgaard, Knudsen and
St\"ockel \cite{Kn17} give an algorithm for $C_{2k}$ detection with running time
$O(m^{2k/(k+1)})$ for every constant $k\geq 2$. This matches the $O(n^2)$
running time from \cite{dense_cycle_detection_YZ} for
$m=\Theta(n^{1+1/k})$ and improves on it for $m<o(n^{1+1/k})$. 

In fine-grained complexity, it is common to assume widely-believed hypotheses and use reductions to obtain conditional lower bounds for fundamental problems. 
For the special case of cycle detection problems, most conditional lower bounds are based on hypotheses related to triangle detection. One of the most common hypotheses is that triangle detection in $n$-node graphs does not have a ``combinatorial'' \footnote{The class of combinatorial algorithms is not
well defined, but intuitively refers to algorithms that do not use fast matrix
multiplication as a subroutine.} $O(n^{3-\eps})$ time algorithm for $\eps>0$ (in the word-RAM model). Vassilevska W. and Williams \cite{focsy} showed that this hypothesis is equivalent to the so called BMM Hypothesis that postulates that there is no $O(n^{3-\eps})$ time combinatorial algorithm for multiplying two $n\times n$ {\em Boolean} matrices.

As mentioned earlier, it is easy to show that under the BMM hypothesis, detecting any odd cycle $C_{2k+1}$ requires $n^{3-o(1)}$ time.
Dahlgaard, Knudsen and St\"ockel
\cite{Kn17} gave lower bounds for combinatorial detection of even cycles of fixed length in {\em sparse} graphs. They show that there is no
combinatorial algorithm for $C_6$ detection, or $C_{2k}$ detection for any
$k>4$, with running time $O(m^{3/2-\eps})$ for $\eps>0$. Lincoln and Vyas \cite{LincolnVyas}
give a similar conditional lower bound under a different hypothesis, extending the lower bounds for potentially non-combinatorial algorithms. Lincoln and
Vyas show that, for large enough constant $k$, a $C_{2k}$ detection algorithm in
graphs with $m\le O(n)$ with running time $m^{3/2-\eps}$ would imply an
algorithm for $\max$-3-SAT with running time $2^{(1-\eps')n} n^{O(1)}$.



The problem of {\em listing} even cycles is less well-understood. Without improving
the $C_{2k}$ detection algorithms discussed above, the best result that we could
hope for in general is an algorithm with running time $O(n^2+t)$, where $t$ is
the number of $2k$-cycles in the graph, and $O(m^{2k/(k+1)}+t)$ in the sparse
setting where $m<o(n^{1+1/k})$. Recently, \cite{Ab22} and \cite{3sumLBCe} gave
such an algorithm for $C_4$ listing. In fact, Jin and Xu's \cite{3sumLBCe} algorithm gives a
stronger guarantee: After $O(m^{4/3})$ pre-processing time, they can
(deterministically) enumerate $4$-cycles with $O(1)$ delay per cycle. 
Jin and Xu \cite{3sumLBCe} (and concurrently by Abboud, Bringmann, and Fischer
\cite{3sumLBAmir}) shows that, under the 3SUM Hypothesis, there is no algorithm
for  $C_4$ enumeration with $O(n^{2-\eps})$ or $O(m^{4/3-\eps})$ pre-processing
time and $n^{o(1)}$ delay, for $\eps>0$.

For $C_6$'s, Jin, Zhou and Vassilevska Williams \cite{C6sCe} give an
$\tilo(n^2+t)$ listing algorithm.
For sparse listing algorithms, the previous state of the art is
\cite{alon1997finding} whose work implies an $O(m^{(2k-1)/k}+t)$ time $C_{2k}$
listing algorithm. See also \cite{bringmannclass} where the complexity of the harder problem of listing $H$-partite\footnote{A $k$-node $H$ is an $H$-partite subgraph of a $k$-partite graph $G$ with vertex set $V=\cup_{a\in V(H)} V_a$ if there are $k$ vertices $v_1,\ldots,v_k$ such that $v_a\in V_a$ for each $a\in \{1,\ldots,k\}$ and the mapping $a\rightarrow v_a$ is an isomorphism between $H$ and the subgraph of $G$ induced by $v_1,\ldots,v_k$. In other words, instead of looking for an arbitrary subgraph of $G$ isomorphic to $H$, one only focuses on the subgraphs with exactly one node in each $V_a$ and such that the node picked from $V_a$ corresponds to node $a$ of $H$.} subgraphs is investigated.

\subsection{Our contributions}
We consider the problem of listing all $C_6$s in an $m$-edge graph.
The best known result so far is an  $O(m^{5/3}+t)$ time algorithm that follows from the work of \cite{alon1997finding}. Our main result is the first improvement over this 27 year old running time:
\begin{restatable*}{theorem}{listfast}\label{thm:listfast}
There is an algorithm for listing $C_6$'s in time $\tilo(m^{1.6}+t)$.
\end{restatable*}


We now summarize the ideas
needed to prove \cref{thm:listfast}.

\cref{thm:listfast} follows easily after establishing a certain bound on the
number of \defn{capped $k$-walks} in a graph.
Capped $k$-walks are a notion introduced by Dahlgaard, Knudsen and St\"ockel in
\cite{Kn17}. Roughly speaking, a capped $k$-walk is a walk of length $k$ where
the first vertex in the walk has higher degree than the remaining vertices in
the walk (handling vertices of equal degree requires a bit of additional care).
\cite{Kn17}'s $O(m^{2k/(k+1)})$ time $C_{2k}$ detection algorithm is based on
the following fact:
\begin{restatable*}{theorem}{thmcappedwalks}\cite{Kn17}
\label{thm:cappedwalks}
Let $G$ be a $C_{2k}$-free graph with maximum degree $\Delta(G)\le m^{2/(k+1)}$. Then, there
are at most $\tilo(m^{2k/(k+1)})$ capped $k$-walks in $G$.
\end{restatable*}
One of our main contributions is a {\em simplified proof} of \cref{thm:cappedwalks},
although with polylogarthmically worse guarantees than \cite{Kn17}. This
simplified analysis of capped $k$-walks is quite helpful in obtaining
\cref{thm:listfast}.
The key lemma used to prove \cref{thm:cappedwalks} is a generalization to
bipartite graphs of Bondy and Simonovits' classic theorem on
the extremal number of $C_{2k}$'s. Specifically, the result is:
\begin{theorem}\cite{Kn17}\label{fact:bipartiteextremal}
Let $G$ be a bipartite graph with vertex parts of sizes $L,R$ and with $m$
edges. If $m>100k(L+R+(LR)^{(k+1)/(2k)})$ then $G$ contains a $C_{2k}$.
\end{theorem}

In order to use capped $k$-walks for a listing algorithm, it would be useful to
have a \defn{supersaturation} variant of \cref{fact:bipartiteextremal}, which
would guarantee the presence of \emph{many} $C_{2k}$'s if the edge density is
large; the fact that this supersaturation result could help with listing was
communicated to us by Jin and Zhou \cite{supersat_observation}. 
The supersaturation analog of Bondy and Simonovits' \cite{evenextremal} extremal
number for $C_{2k}$'s is known:
\begin{theorem}\cite{JiangYep20}\label{fact:vanillasupersat}
For every integer $k\geq 2$, there exist constants $c,C$ such that if $G$ is an $n$-node graph with $m\geq Cn^{1+1/k}$ edges, then $G$ contains at least $c(m/n)^{2k}$ copies of $C_{2k}$.
\end{theorem}

Jin and Zhou \cite{supersat_observation} formulated the following conjectured generalization of
\cref{fact:vanillasupersat}:

\begin{restatable*}{conjecture}{supersatconj} 
\label{conj:supersat}
\textbf{The Unbalanced Supersaturation Conjecture [Jin and Zhou'24]:}
Let $G$ be an $m$ edge bipartite graph with vertex bipartition $A \sqcup B$,
with $t$ $C_{2k}$'s.
Suppose $m \ge 100k(|A|+|B| + (|A||B|)^{(k+1)/2k} )$.
Then,
\[ t \ge \Omega\left( \frac{m^{2k}}{|A|^{k}|B|^{k}} \right). \] 
\end{restatable*}

Jin and Zhou obtained the following conclusion using the approach of \cite{Kn17}:

\begin{fact}[\cite{supersat_observation}]
If the Unbalanced Supersaturation Conjecture is true, then there is an $\tilo(t+m^{2k/(k+1)})$ time algorithm for $2k$-cycle listing for all $k\geq 2$.
\end{fact}

In \cref{thm:conditional_listing} we give a simple proof of Jin and Zhou's fact above.
Thus, an approach to obtaining the conjectured optimal running time of $\tilo(t+m^{2k/(k+1)})$ for $C_{2k}$-listing would be to prove \cref{conj:supersat}. Unfortunately,
establishing or refuting \cref{conj:supersat} remains a challenging open problem.

Our main technical contribution is a proof of a weaker version of
\cref{conj:supersat} for $k=3$. Then, we show that this partial progress towards
\cref{conj:supersat} can be used in our simplified method for analyzing capped
$k$-walks to obtain a bound on the number of capped $k$-walks, and
consequentially an improved $C_6$ listing algorithm in the sparse setting
(namely, \cref{thm:listfast}).

\subsection{Open Questions}
Proving or refuting \cref{conj:supersat} is an important open question.
It also is valuable to consider other avenues towards obtaining
$C_{2k}$ listing algorithms, especially in case \cref{conj:supersat} turns out
to be false.
The listing algorithms of \cite{C6sCe}, \cite{3sumLBCe}, \cite{Ab22} use a
variety of combinatorial insights that could potentially be generalized to
larger $k$.
It is also possible that a hybrid approach is productive: one can first use
progress towards proving \cref{conj:supersat} to force the instance to have a
specific structure, and then use different combinatorial insights to solve the
structured version of the problem.

\subsection{Paper Outline}
In \cref{sec:simple_capped} we present our simplified analysis of \cite{Kn17}'s
$C_{2k}$ detection algorithm. 
In \cref{sec:listing_with_supersat} we show how to
modify our simple analysis of \cite{Kn17}'s $C_{2k}$ detection algorithm to get
a $C_{2k}$ listing algorithm, assuming the Unbalanced Supersaturation
Conjecture. In \cref{sec:supsersat_progress} we present progress towards
resolving the Unbalanced Supersaturation Conjecture. In
\cref{sec:listing_progress} we show how to use our progress towards
\cref{conj:supersat} in our simplified capped $k$-walk analysis to obtain a
listing algorithm for $C_6$'s with running time $\tilo(m^{1.6}+t)$.

\subsection{Notations}
We use $\Delta(G)$ to denote the maximum degree of graph $G$. 
We write $N(v)$ to denote the neighborhood of vertex $v$, and for $S\subseteq
V(G)$ we write $N_S(v)$ to denote $N(v)\cap S$. 
For graph $G$, we will use $V(G),E(G)$ to denote the vertex and edge sets of $G$. When $G$ is clear from the context we also denote $V(G)$ by $V$ and $E(G)$ by $E$. We let
$n=|V|,m=|E|$ (when the graph being discussed is clear), and assume $m\ge n$. 
For vertex subsets  $A,B$ we write $e(A,B)$ to denote $|E\cap (A\times B)|$.
A walk / path of length $k$ will refer to a walk / path with $k$ edges.
Given $A,B\subseteq V$, we will write $G[A,B]$ to denote the induced subgraph on
$A,B$, i.e., a graph with vertex set $A\cup B$ and edge set $E(G)\cap (A\times B)$.
We will write $G[A]$ to denote $G[A,A]$.

We write $[x]$ to denote the set $\set{1,2,\ldots, \ceil{x}}$. 
We use $\log$ to denote
the base-$2$ logarithm. Define the \defn{dyadic intervals} $\mc{I}_n$ as
follows: letting $n'$ denote the power of two in $[n,2n)$, we define $\mc{I}_n =
(\set{0},\set{1}, [2,4), [4,8), \ldots, [n'/4,n'/2), [n'/2,n'])$. Note that $|\mc{I}_n|
= 1+\log n'$. For $j\in [\log n']$, the $j$-th dyadic interval refers to the
$j$-th set in the list $\mc{I}_n$ of dyadic intervals, \emph{starting counting
from $\set{1}$}, or in other words \emph{excluding the set $\set{0}$}. That is, the $j$-th dyadic interval is $[2^{j-1},2^j)$ for $j\in [(\log n')-1]$ and the $(\log n')$-th one is $[n'/2,n']$.

\section{A Simpler Analysis of the DKS $C_{2k}$ Detection Algorithm}\label{sec:simple_capped}
Dahlgaard, Knudsen and St\"ockel \cite{Kn17} define \defn{capped $k$-walks}
and give a sophisticated analysis to bound the number of capped $k$-walks in
$C_{2k}$-free graphs (assuming $\Delta(G)\le m^{2/(k+1)}$). In this section we
provide a simpler --- although quantitatively worse by logarithmic factors ---
version of their analysis. This will be useful in future sections when we 
analyze the number of capped $k$-walks in graphs that are not $C_{2k}$-free.

Given a total ordering $\succ$ of the vertices in a graph, a
\defn{$\succ$-capped $k$-walk} \cite{Kn17} is a walk $x_0, x_1, \ldots, x_k$
where $x_0\succ  x_i$ for all $i\in [k]$. 
Throughout the paper we will assume implicitly an ordering $\succ$ of the
vertices that satisfies $v_i\succ v_j$ whenever $\deg(v_i)> \deg(v_j)$ (but is otherwise arbitrary). Thus, we will refer simply to capped $k$-walks (leaving
the dependence on $\succ$ implicit).
The main result of this section is a simplified proof of the following fact:

\thmcappedwalks
We note that if $G$ were $(m/n)$-regular then it would contain $\Theta(n(m/n)^k)
\le \O(m^{2k/(k+1)})$ $k$-walks due to $m\le \O(n^{1+1/k})$ (which holds because
$G$ is $C_{2k}$-free). While a non-regular
graph $G$ may have more than $\O(m^{2k/(k+1)})$ many $k$-walks,
\cref{thm:cappedwalks} states
that $G$ does not have more than this many \emph{capped} $k$-walks
(assuming $\Delta(G)\le m^{2/(k+1)}$). 
At a high level, the proof of the bound on capped $k$-walks follows from
applying the bipartite generalization of Bondy and Simonovits' extremal number
bound \cite{evenextremal} (\cref{fact:bipartiteextremal}) to a certain
\defn{layered organization} of the graph with nice regularity properties.
This layered organization of the graph is our contribution. 

In \cite{Kn17}, the
authors took a different approach: they defined the ``$\phi$-norm'' of a vector
$v$ as
\[ ||v||_\phi = \int_0^{\infty}\sqrt{|\setof{i}{|v_i|\ge x}|} dx. \] 
They then related the $\phi$-norm of $X_G^{k}\one$, (where $X_G$ is the graph's
adjacency matrix and $\one$ is the all-ones vector) to the number of capped
$k$-walks in $G$, and analyzed a large set of inequalities using combinatorial
insights to bound this $\phi$-norm. 
We find that the layered organization of the graph technique is more helpful
than the $\phi$-norm approach in elucidating \cite{Kn17}'s elegant result, and
that this layered organization technique is easier to use in bounding the number
of capped $k$-walks in graphs which are not $C_{2k}$-free.
Our layered organization of the graph is depicted in \cref{fig:buckets}, and
formally defined and analyzed in \cref{lem:DKSbuckets}.
\begin{figure}[h]
    \centering
    \includegraphics[width=.8\linewidth]{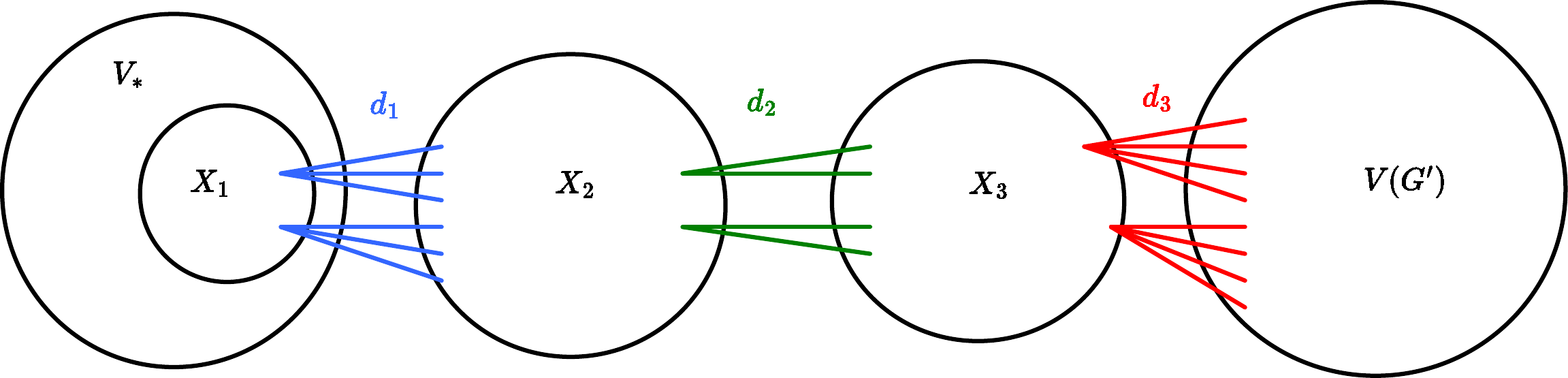}
    \caption{A layer pattern capturing many capped $3$-walks.}
\label{fig:buckets}
\end{figure}

\begin{lemma}
\label{lem:DKSbuckets}
Fix graph $G$. There exists $d_*\in [n]$ such that if we let $V_*$ denote the
set of vertices with degree between $[d_*/2, d_*]$ and let $G'$ denote the
induced subgraph of $G$ on vertices with degree at most $d_*$, then 
there 
exist $(X_1,\ldots, X_k)\subseteq V_*\times V(G')^{k-1}, (d_1,\ldots,
d_k)\in \N^k$ such that letting $X_{k+1}=V(G')$,
\begin{enumerate}
\item The number of capped $k$-walks in $G$ is at most $O(\log^{k+1} n)$ times 
the number of $k$-walks in $X_1\times X_2\times \cdots \times X_k\times V(G')$, 
\item For each $i\in [k]$, every $v\in X_i$ has $|N_{X_{i+1}}(v)|\in [d_i/2, d_i]$.
\end{enumerate}
\end{lemma}
\begin{proof}
For $j\in [\log n]$, let $V_j$ denote the set of vertices $v\in V$ with degree
lying in the $j$-th dyadic interval.
Fix $j$ maximizing the number of capped $k$-walks in $G$ that start in a vertex
in $V_j$. Let $V_*=V_j, d_* = 2^{j}$, 
and let $G'$ be the induced subgraph
on vertices of degree at most $d_*$.
The number of capped $k$-walks in $G$ is at most $O(\log n)$ times the number of
$k$-walks in $G'$ that start in $V_*$; this is because the vertices in
$V(G)\setminus V(G')$ are of too high degree to participate in a capped $k$-walk
starting in $V_*$, and because every capped $k$-walk must start in some $V_j$.


For each $i_k\in [\log n]$ let $W^{(k)}_{i_k}$ denote the set of vertices in
$G'$ whose degree in $G'$ lies in the $i_k$-th dyadic interval.
Now, iteratively for $\ell=k-1, k-2, \ldots, 2$, for each $i_\ell \in [\log n]$ define 
$W^{(\ell)}_{i_k, i_{k-1}, \ldots, i_\ell}$ to be the set of vertices $v\in
V(G')$ such that the number of neighbors of $v$ in $W^{(\ell+1)}_{i_k,
i_{k-1},\ldots, i_{\ell+1}}$ lies in the $i_\ell$-th dyadic interval.
Finally, for each $i_1 \in [\log n]$ define 
$W^{(1)}_{i_k, i_{k-1}, \ldots, i_1}$ to be the set of vertices $v\in V_*$ such
that the number of neighbors of $v$ in $W^{(2)}_{i_k, i_{k-1},\ldots, i_{2}}$
lies in the $i_1$-th dyadic interval.

Intuitively, looking at Figure \ref{fig:buckets}, we are building up the sets $X_\ell$ from right to left. Here we think of $W^{(\ell)}_{i_k, i_{k-1}, \ldots, i_\ell}$ as representing $X_\ell$ for some choice of $i_k, i_{k-1}, \ldots, i_\ell$ that we will fix shortly.

Now, fix a choice $(i_k,\ldots,i_1)\in [\log n]^{k}$ that maximizes the number of $k$-walks  
\begin{equation}
  \label{eq:fancyform}
(x_1,\ldots, x_{k+1})\in W^{(1)}_{i_k,\ldots, i_1} \times
W^{(2)}_{i_{k},\ldots, i_2}\times \ldots, W^{(k)}_{i_k}\times V(G').
\end{equation}
Note that every $k$-walk in $G'$ starting in $V_*$ must be of the form
described in \cref{eq:fancyform} for \emph{some} $(i_k,\ldots,i_1)\in [\log n]^{k}$.
Thus, for our choice of $(i_k,\ldots,i_1)$, at least an 
$\Omega(1/\log^{k}n)$-fraction of the $k$-walks in $G'$ starting in $V_*$ are of
the form \cref{eq:fancyform}.
For this choice of $(i_k,\ldots,i_1)$, we define for each $j\in [\log n]$,
\[X_j = W^{(j)}_{i_k,\ldots, i_j}, \quad d_j = 2^{i_j}.\]
These are the sets $X_j,d_j$ that we wanted in the lemma statement, and they
satisfy the required degree regularity property. Also because an 
$\Omega(1/\log^{k} n)$-fraction of the $k$-walks in $G'$ starting in $V_*$ are
of the form \cref{eq:fancyform}, and because the number of capped
$k$-walks in $G$ is at most $O(\log n)$ times the number of $k$-walks in $G'$
starting in $V_*$, we have that the number of capped $k$-walks in $G$ is at most
$\tilo(1)$ times the number of $k$-walks of the form \cref{eq:fancyform}.

\end{proof}

We have shown in \cref{lem:DKSbuckets} that in order to bound the number of
capped-$k$ walks in $G$, it suffices to bound the number of $k$-walks where each
vertex of the $k$-walk is required to belong to a specific ``layer'' set, where
consecutive layers satisfy a degree regularity property.
This motivates studying how large an edge density we can have between two such
consecutive layers while avoiding $C_{2k}$'s. 
First, we need a minor modification of \cref{fact:bipartiteextremal}.
\begin{corollary}\label{cor:non-disjoint-extreme}
Let $G$ be a $C_{2k}$-free graph. Fix $A,B\subseteq V(G)$ \emph{not necessarily
disjoint}. Then, 
\[e(A,B) \le O(|A|+|B|+(|A||B|)^{(k+1)/(2k)}).\]
\end{corollary}
\begin{proof}
A well known fact (see e.g., \cite{zhao2023graph}) in graph theory is that every
$m$-edge graph has a bipartite subgraph with at least $m/2$ edges. 
Here we can apply the same idea. Let $C=A\cap B, A'=A\setminus C, B'=B\setminus C$. Form a partition of $G$ as follows. Initialize $L=A',R=B'$ and then for every $c\in C$ place $c$ in $L$ with probability $1/2$ and in $R$ otherwise. Then remove all edges not in $L\times R$. In expectation, at least half of the edges in $A\times B$ are preserved. Hence some setting of the random bits gives a bipartite subgraph with at least half the edges from $A\times B$ on which we can apply \cref{fact:bipartiteextremal} to upper-bound $e(A,B)/2$ by  $O(|L|+|R|+(|L||R|)^{(k+1)/(2k)})$ which is the desired quantity as $|L|\leq |A|$ and $|R|\leq |B|$.
\end{proof}
Now we prove a lemma to explain how consecutive layers in the
decomposition of \cref{lem:DKSbuckets} interact.
\begin{lemma} \label{lem:keylem}
Let $G$ be a $C_{2k}$-free graph with $\Delta(G)\le m^{2/(k+1)}$. Fix
$A\subseteq V(G)$ and $d\le \Delta(G)$. Let $B$ denote the set of
vertices $v\in V(G)$ with $|N_A(v)|\in [d/2, d]$. Then,
\[ d\sqrt{|B|}\le \O(m^{\frac{1}{k+1}}\sqrt{|A|}). \] 
\end{lemma}
\begin{proof}
Define $a=|A|, b=|B|, \Delta=\Delta(G)$. 
In this proof we will write $x\les y$ to denote $x\le O(y)$. 
We aim to show $d^2 b/a \les m^{2/(k+1)}$. 
This will follow from a simple chain of inequalities, with the main
combinatorial ingredient being \cref{cor:non-disjoint-extreme}.
In the inequalities we will use $\one_{[P]}$ to denote the indicator of whether
$P$ is true (i.e., it evaluates to $1$ or  $0$ based on the truth of proposition
$P$).

We have:
\begin{align*}
  d^2b/a & = \frac{(db)^2}{ab}\\
  &\les \min\left( \frac{d^2b}{a}, \frac{e(A, B)^2}{ab}\right)
&& \text{$B$ to $A$ degree regularity}\\
&\les \min\left( \frac{d^2b}{a}, \frac{a^2+b^2+(a b)^{1+1/k}}{ab} \right)
&& \text{\cref{cor:non-disjoint-extreme}}\\
&\les \min\left( \frac{d^2b}{a}, \frac{a}{b}\right) + \frac{b}{a}+ \min\left(
\frac{d^2b}{a},  (ab)^{1/k} \right)
&& \text{re-arrange terms}\\
&\les \sqrt{ \frac{d^2b}{a}\cdot \frac{a}{b}} + \frac{b}{a}+ \min\left(
\frac{d^2b}{a},  (ab)^{1/k} \right)
&& \min(x,y)\le\sqrt{xy} \\
&\les d + \frac{b}{a}+ \min\left( \frac{d^2b}{a},  (ab)^{1/k} \right)
&& \text{simplify}\\
&\les \Delta + \frac{b}{a}+ \min\left( \frac{d^2b}{a},  (ab)^{1/k} \right)
&& \text{$d\le \Delta$}\\
&\les \Delta + \frac{\Delta\cdot a}{a}+ \min\left( \frac{d^2b}{a},  (ab)^{1/k} \right)
&& b\le\abs{\bigcup_{v\in A}N(v)}\le \Delta\cdot a \\
&\les m^{2/(k+1)} + \min\left( \frac{d^2b}{a},  (ab)^{1/k} \right)
&& \Delta \le m^{2/(k+1)}\\
&\les m^{2/(k+1)} + \one_{[a>d m^{1-2/(k+1)}]}\cdot\frac{d^2b}{a}+\one_{[a\le
dm^{1-2/(k+1)}]}\cdot (ab)^{1/k} 
&& \text{split $\min$ }\\
&\les m^{2/(k+1)} + \frac{d^2b}{dm^{1-2/(k+1)}}+(dm^{1-2/(k+1)}b)^{1/k} 
&& \text{monotonicity}\\
&\les m^{2/(k+1)} + \frac{m}{m^{1-2/(k+1)}}+(m^{2-2/(k+1)})^{1/k} 
&& db \le O(m)\\
&\les m^{2/(k+1)} 
&& \text{simplify.}\\
\end{align*}
\end{proof}

Now we combine \cref{lem:DKSbuckets} and \cref{lem:keylem} to obtain \cref{thm:cappedwalks}.
\begin{proof}[Proof of \cref{thm:cappedwalks}]
We apply \cref{lem:DKSbuckets} to obtain $V_*,d_*, G'$, $X_1,\ldots,
X_k\subseteq V(G'), d_1,\ldots, d_k$ with the properties described in
\cref{lem:DKSbuckets}. We can easily count the number of $k$-walks in $X_1\times \cdots \times
X_k\times V(G')$ due to the regularity property that consecutive $X_i$'s enjoy;
it is at most
\begin{equation}\label{x1proddj}
|X_1| \cdot \prod_{j=1}^k d_j. 
\end{equation}
We now bound \cref{x1proddj} using \cref{lem:keylem}, which implies that for
each $i\in[k-1]$,
\begin{equation}
  \label{eq:adjacentbound}
d_i\sqrt{|X_i|} \le \O(m^{1/(k+1)} \sqrt{|X_{i+1}|}).
\end{equation}
Repeatedly applying \cref{eq:adjacentbound} to \cref{x1proddj} gives:
\begin{equation}
  \label{eq:prodexpand}
|X_1| \prod_{j=1}^{k} d_j \le \sqrt{|X_1| |X_k|} d_k \cdot \O(m^{(k-1)/(k+1)}). 
\end{equation}
To bound \cref{eq:prodexpand} we make two observations.
First, note that $|X_k|d_k \le O(m)$ because all $v\in X_k$ have degree at least
$d_k/2$.
Second, recall that $X_1\subseteq V_*$, and $d_* \ge d_k$. Thus, 
$|X_1|d_k \le |X_1|d_* \le O(m)$.
Applying this in \cref{eq:prodexpand} gives
\begin{equation}
  \label{eq:x1djm2k}
|X_1|\prod_{j=1}^{k} d_j \le \O(m^{2k/(k+1)}).
\end{equation}
By \cref{lem:DKSbuckets}, the number of capped $k$-walks in $G$ is at most
$\tilo(1)$ times  \cref{eq:x1djm2k}; this is the desired bound.
\end{proof}

\cite{Kn17} showed that \cref{thm:cappedwalks} implies a $C_{2k}$ detection algorithm. 
The $C_{2k}$ detection algorithm is based on the following slightly more general
lemma; this lemma is also the basis of our listing algorithm.
\begin{lemma}\label{lem:listlist}
Fix parameter $\Delta\in \N$, and graph $G$. 
Let $W$ be the number of capped $k$-walks in $G$ that start at a vertex with degree at most $\Delta$.
There is a (randomized)
algorithm to list the $C_{2k}$'s in $G$ 
in time $\tilo(m^2/\Delta + W + t).$
\end{lemma}
\begin{proof}
Label the vertices $v_1,\ldots, v_n$, where $v_i\succ v_j$ iff $i>j$. For each
$i\in [n]$, define a subgraph $G_i$ on vertices $\set{v_1,\ldots, v_i}$,
consisting of the edges that are part of the $k$-step BFS out of vertex $v_i$.
For $i=1,2,\ldots, n$, we will list the $C_{2k}$'s in $G_i$ that involve $v_i$.
Note that this will result in listing all $C_{2k}$'s; in particular,
if $v_{i}$ is the largest (by $\succ$)
vertex in some $C_{2k}$, then we will list this $C_{2k}$ on iteration $i$ of our procedure.
Let $t_i$ denote the number of $C_{2k}$'s in $G_i$ that involve $v_i$; we now
discuss how to list these $C_{2k}$'s.
To list the $C_{2k}$'s in $G_i$ involving $v_i$, we use color coding. We color
each vertex with a random color from $[2k]$. Then we restrict our search to
\defn{colorful} $C_{2k}$'s: $C_{2k}$'s with a vertex of every color. 
Listing the colorful $C_{2k}$'s involving $v_i$ can be done in time $O(|E(G_i)|+t_i)$.
Each $C_{2k}$ becomes colorful with probability $2^{-2k}\ge \Omega(1)$, so by
repeating this $\tilo(1)$ times, we will list all the $C_{2k}$'s in $G_i$
involving $v_i$, with high probability.

Now we analyze the cost of our algorithm. 
Note that $|E(G_i)|$ is at most the number of capped $k$-walks that start at
$v_i$.
Let $T$ be the largest $i$ such that $\deg(v_i)\le \Delta$.
Then, by definition
\[\sum_{i=1}^T |E(G_i)| \le O(W).\]
There are at most $m/\Delta$ vertices of degree larger than $\Delta$.
Thus, 
\[\sum_{i=T+1}^n |E(G_i)| \le O(m^2/\Delta).\]
Finally, we have $\sum_i t_i = O(t)$.
Combined, we see that the algorithm runs in time $\tilo(t+W+m^2/\Delta)$.
\end{proof}

\begin{corollary}\label{cor:detection}
There is a (randomized) $\tilo(m^{2k/(k+1)})$ time algorithm for $C_{2k}$ detection.
\end{corollary}
\begin{proof}
This follows by applying \cref{lem:listlist} with $\Delta=m^{2/(k+1)}$ (and
terminating the algorithm if it runs for too long).
\end{proof}

 \cref{lem:listlist} and \cref{cor:detection} can be derandomized using standard techniques (perfect hash families etc) \cite{colorcoding}.


\newcommand{\hex}{\mathsf{hex}}
\section{Progress Towards the Supersaturation Conjecture}\label{sec:supsersat_progress}
Jin and Zhou observed (\cite{supersat_observation}) that the unbalanced
supersaturation conjecture implies sparse listing algorithms for even cycles of
any length; we give a simple proof of this in \cref{sec:listing_with_supersat}.
In this section we present partial progress towards proving the unbalanced
supersaturation conjecture. We
restrict our attention to \defn{hexagons} ($C_6$'s) because this is the
smallest case where sparse listing algorithms are not known. In
\cref{sec:listing_progress} we will show that our partial progress towards the supersaturation
conjecture can be translated into partial progress towards listing hexagons in
time $O(m^{1.5}+t)$. 
Before stating our result, we state some simple folklore facts about path
supersaturation (see \cref{appendix:path-facts} for proofs). 
\begin{fact}\label{fact:p2}
Let $G$ be a bipartite graph with parts $A,B$ of sizes $L,R$. If $m\ge 2R$, then
the number of $2$-paths in $A\times B\times A$ is at least $L(m/L)(m/R)/2$.
\end{fact}
\begin{fact}\label{fact:p4}
Let $G$ be a bipartite graph with parts $A,B$ of sizes $L,R$. If $m\ge 50(R+L)$, then
the number of $4$-paths in $A\times B\times A\times B\times A$ is at least
$\Omega(L(m/L)^2(m/R)^2)$.
\end{fact}
Our partial supersaturation result for hexagons is as follows:
\begin{theorem}\label{thm:partialsupersat}
Let $G$ be a bipartite graph with vertex bipartition $A\sqcup B$ and $t$
hexagons. Let $L=|A|, R=|B|$ with $100\le L\le R$, and let $d_L=m/L, d_R=m/R$. Suppose 
$m/\ceil{\log(L+R)}^{10}>L+R+(LR)^{2/3}$. 
Then, there is an absolute constant $c>0$ such that 
\[ t \ge d_R^3 d_L^3 \cdot \min(1, d_R^3/L, d_L^2d_R^2/L^2)\cdot c/\log^{70} n. \] 
\end{theorem}
\begin{proof}
Define $n=2^{\ceil{\log (L+R)}}$. The constraint $100(\log n)^{10}\le m\le n^2$
implies $n\ge 10^{8}$.

\paragraph{Organizing the Graph.}\\
Our approach is to first organize the graph and then find hexagons via two and
three step BFS's. The start of our organization of the graph is the following
claim:
\begin{claim}\label{clm:somestruct}
There exist disjoint sets $A',A''\subseteq A$ and $B',B''\subseteq B$, of
sizes $L'=|A'|,L''=|A''|,R'=|B'|,R''=|B''|$, and there exist values
$d_L',d_L'',d_R'$ such that:
\begin{itemize}
\item Vertices $v\in A'$ have $|N_{B'}(v)|\in [d_L'/4,d_L']$ and $e(A',B') \ge m/\log^{4} n$.
\item Vertices $v\in B'$ have $|N_{A''}(v)|\in [d_R'/4, d_R']$ and
$e(B',A'')\ge m/\log^{4}n$.
\item Vertices $v\in A''$ have $|N_{B''}(v)|\in [d_L''/4, d_L'']$ and $e(A'',B'')
\ge m/\log^{4}n$.
\end{itemize}
\end{claim}
\begin{proof}
Define $Y_1=Y_3=A,Y_2=Y_4=B$.
Form vertex subsets $X_1,X_2,X_3,X_4$ as follows.
Define $X_4=Y_4$.
For $i=3,2,1$ define $X_i$ as follows:
Partition $y\in Y_i$ based on which dyadic interval $|N_{X_{i+1}}(y)|$ falls in.
Let $X_i$ be a part in this partition with the maximum number of edges to
$X_{i+1}$. The number of parts in our dyadic partition (disregarding the
$\set{0}$ part, which does not contribute edges) is $\log n$. Thus, 
\begin{equation}
  \label{eq:bucketingguarantees}
e(X_4,X_3)\ge m/\log n, e(X_3,X_2)\ge m/\log^2 n, e(X_2,X_1)\ge m/\log^3 n. 
\end{equation}
Now, assign each vertex $v\in V(G)$ a uniformly random color from $\set{0,1}$.
For $i=1,2$ define $X_i'$ to be the subset of $X_i$ colored $0$.
For $i=3,4$ define $X_i'$ to be the subset of $X_i$ colored $1$.
Note that the sets $X_1',X_2',X_3',X_4'$ are disjoint by construction: in particular, even though $X_2\subseteq X_4$, by the coloring $X'_2\cap X_4'=\emptyset$.

We now argue that there exists some coloring such that the sets
$\setof{X_i'}{i\in[4]}$ have similar sizes and degree regularity properties to
the sets $\setof{X_i}{i\in[4]}$. First, because for each $i$, $|X_i|\ge 100$,
with probability at least $.9999$ we have $|X_i'|\ge |X_i|/4$ for each $i$.
Now, fix $i\in [3], v\in X_i$. The value $|N_{X_{i+1}'}(v)|$ is is a Bernoulli
random variable with $|N_{X_{i+1}}(v)|$ trials, and $p=1/2$ success probability
for each trial. By a Chernoff bound,
\begin{equation}
  \label{eq:chernoff}
\Pr[|N_{X_{i+1}'}(v)|\le |N_{X_{i+1}}(v)|/4] \le \exp(-|N_{X_{i+1}}(v)|/16).
\end{equation}
Now, we will argue that $|N_{X_{i+1}}(v)|$ is large enough that we can afford to
union bound across \cref{eq:chernoff} for all $i\in[3],v\in X_i$.
For $i=1,2,3$ choose $d_i$ such that each vertex $v\in X_i$ has
$|N_{X_{i+1}}(v)|\in[d_i/2,d_i]$; the sets $X_i$ were constructed so that this is
possible.
Now, \cref{eq:bucketingguarantees} implies for $i\in[3]$ that
\begin{equation}
  \label{eq:showNbig}
d_i \ge \frac{m}{|X_i| \log^3 n} \ge \log^{7} n.
\end{equation}
Therefore, the probability that \cref{eq:chernoff} fails for any $i\in [3], v\in
X_i$ is at most
\[ 3n\exp(-(\log^{7}n)/32) \le .001. \] 
Thus with probability at least $.99$, for each $i\in [4]$, $|X_i'|\ge |X_i|/4$,
and if $i\neq 4$ then each vertex $v\in X_i'$ has $|N_{X_{i+1}'}(v)|\in
[d_i/4,d_i]$, and finally (using \cref{eq:bucketingguarantees})
\[ |X_i'|d_i \ge e(X_i', X_{i+1}') \ge \frac{1}{16}e(X_i, X_{i+1})\ge m/\log^{4}n. \] 
Thus, by the probabilistic method there is some coloring that gives these
properties.
Fix this coloring, and let $A'=X_1', A''=X_3', B'=X_2',B''=X_4'$.
These are precisely the sets we wanted to show exist.

\end{proof}

We continue to use the sets $A',B',A'',B''$ introduced in \cref{clm:somestruct}
throughout the proof, along with the values $L',R',L'',R''$ and
$d_L',d_L'',d_R'$.
A useful immediate corollary of \cref{clm:somestruct} is 

\begin{corollary} \label{eq:primes-bigger}
$d_L'\ge (L/L')d_L/\log^{4}n\ge d_L/\log^{4}n$, 
$d_R'\ge (R/R')d_R/\log^{4}n\ge d_R/\log^{4}n$, and\\
$d_L''\ge (L/L'')d_L/\log^{4}n \ge d_L/\log^{4}n.$
\end{corollary}


Now we further organize the vertices.
In what follows, we will partition vertices based on the \defn{approximate
value} of some quantity associated with the vertex, by which we mean which of
the dyadic intervals $\mc{I}_n$ the quantity lies in. 
For vertex $v$ and vertex subset $X$, define $N^2_X(v)$ to be the set of
vertices $w\in X$ that are distance $2$ from $v$.

Fix vertex $v\in A'$.
Define $(X_v, x_v, \delta_v)$ (visualized in \cref{fig:Xv}) as follows:
Partition $w\in N^2_{A''}(v)$ into \defn{buckets} based on the approximate
value of $|N_{B'}(v)\cap N_{B'}(w)|$.
Let $X_v\subseteq A''$ be a bucket which maximizes the number of edges between
$X_v$ and $N_{B'}(v)$. 
Let $x_v = |X_v|$, and choose $\delta_v\in \set{1}\cup\set{4,8,16\ldots,n}$ such that
each $w\in X_v$ has $|N_{B'}(v)|\in [\delta_v/2, \delta_v]$ (this is possible by
the definition of $\mc{I}_n$, and the fact that the $\set{0}$ bucket does not
contribute any edges).
Then,
\begin{equation}
  \label{eq:deltavxv}
\delta_v x_v \ge e(X_v, N_{B'}(v))\ge \frac{1}{\log n}e(N_{B'}(v),A'') \ge \frac{1}{16\log n} \cdot d_L' d_R'. 
\end{equation}
\begin{figure}[htpb]
  \centering
  \includegraphics[width=0.5\textwidth]{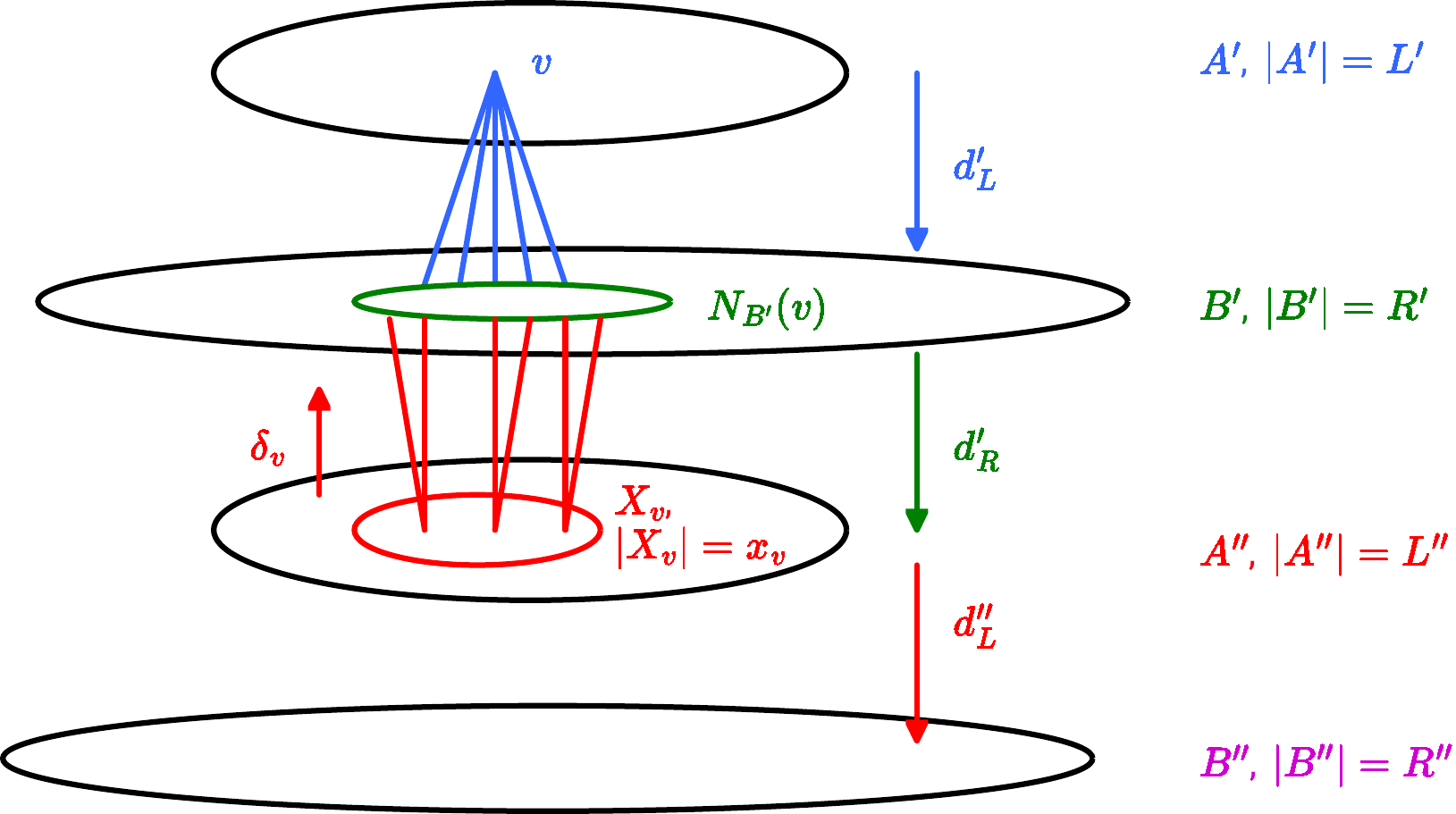}
\caption{Organizing the graph}
\label{fig:Xv}
\end{figure}

Next, we assign each vertex $r\in N_{B'}(v)$ a tuple
$(Y_{rv},y_{rv},\lambda_{rv})$ (visualized in \cref{fig:Yv}) as follows:
Partition $z\in N^2_{B''}(r)$ into buckets based on the approximate value of
$|N_{X_v}(r)\cap N_{X_v}(z)|$.
Let $Y_{rv}\subseteq B''$ be a bucket which maximizes the number of edges
between $Y_{rv}$ and $N_{X_v}(r)$. Let $y_{rv}=|Y_{rv}|$ and choose
$\lambda_{rv}\in \set{1}\cup\set{4,8,16\ldots,n}$ such that each $z\in B''$ has
$|N_{X_v}(r)\cap N_{X_v}(z)| \in [\lambda_{rv}/2,\lambda_{rv}]$.
Then,
\begin{equation}
  \label{eq:lambdarvyrv}
\lambda_{rv} y_{rv}\ge e(N_{X_v}(r), Y_{rv}) \ge \frac{1}{\log n} e(N_{X_v}(r), B''). 
\end{equation}

Finally, partition $r\in N_{B'}(v)$ into buckets based on the
approximate values of $y_{rv}, \lambda_{rv},|N_{X_v}(r)|$.
Let $B_v'\subseteq B'$ be a bucket which maximizes the number of edges between $B_v'$ and
$X_v$. Choose $y_v,\lambda_v,\beta_v$ (visualized in \cref{fig:Yv}) such that each $r\in B_v'$ has $y_{rv}\in
[y_v/2,y_v]$, $\lambda_{rv}\in [\lambda_v/2,\lambda_v]$, and $|N_{X_v}(r)| \in
[\beta_v/2, \beta_v]$.
Then, 
\begin{equation}
  \label{eq:BvprimeXv}
e(B_v', X_v) \ge \frac{1}{\log^3 n} e(N_{B'}(v), X_v).
\end{equation}


\begin{figure}[htpb]
  \centering
  \begin{subfigure}[b]{0.45\textwidth}
    \includegraphics[width=\textwidth]{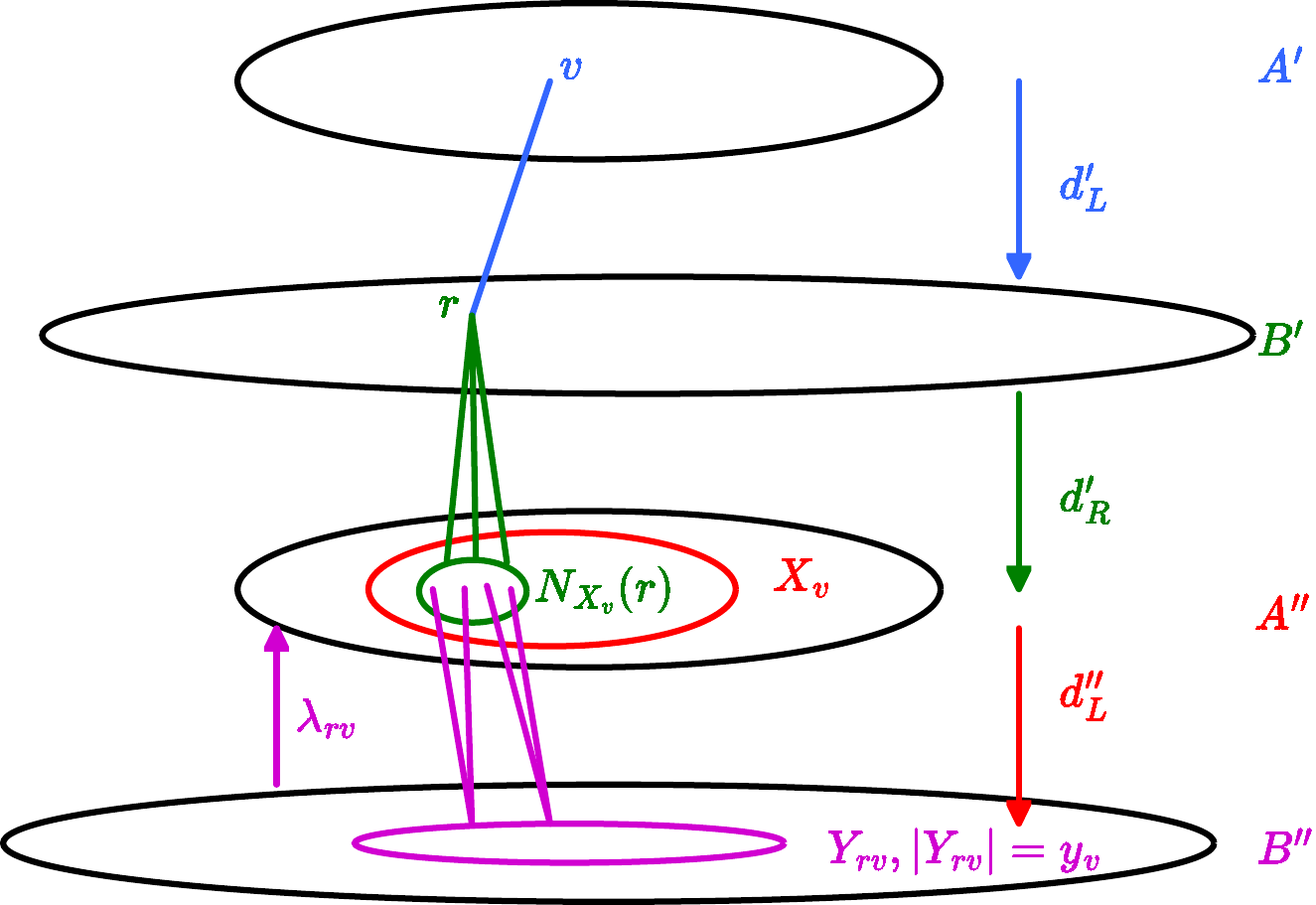}
  \end{subfigure}
  \begin{subfigure}[b]{0.45\textwidth}
    \includegraphics[width=\textwidth]{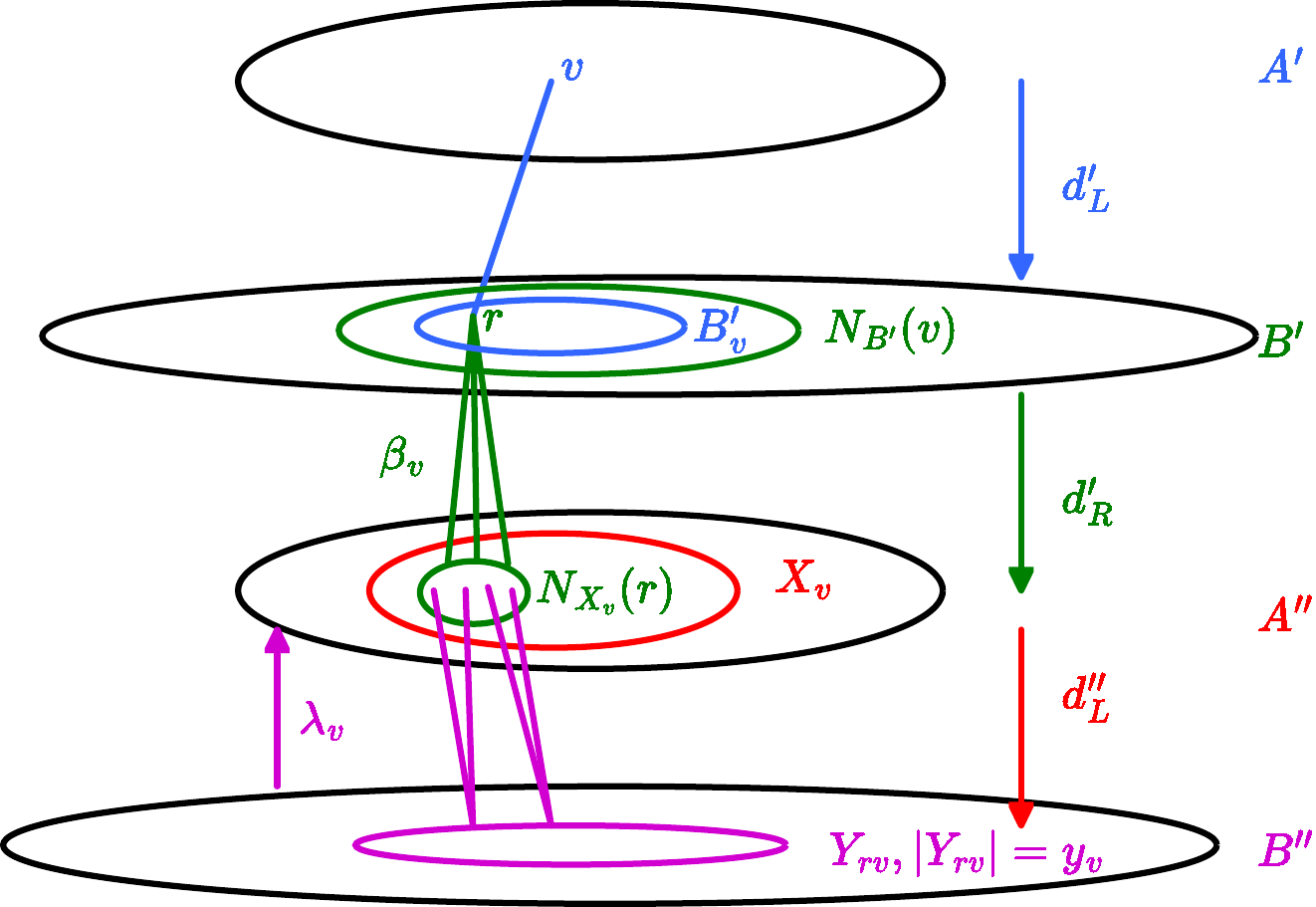}
  \end{subfigure}
\caption{Organizing the graph}
\label{fig:Yv}
\end{figure}

\paragraph{Counting Hexagons in the Organized Graph.}\\
Let $\hex(G)$ denote the number of hexagons in $G$. With the graph nicely
organized we are prepared to bound $\hex(G)$. For vertex $v\in A'$, let
$\hex(v)$ denote the number of hexagons in $G$ where $v$ is the only $A'$ vertex
participating in the hexagon. One natural way of bounding $\hex(G)$ is by
$\hex(G)\ge \sum_{v\in A'}\hex(v)$. We will split into cases based on the
characteristics of the 3-step BFS out of each vertex $v$ and attempt to show
that $\hex(v)$ is large. However, in one of the cases we will show $\hex(G)$ is
large directly from the 3-step BFS out of a single vertex, rather than bounding
$\hex(v)$. We now begin to consider 3-step BFS's based on their characteristics.
See \cref{fig:case12} and \cref{fig:case34} for an illustration of our approach.



\begin{claim}\label{clm:case1}
Suppose vertex $v\in A'$ has $x_vd_L''/4 > 2R$ and $\delta_v\neq 1$. Then
$\hex(v)\ge \Omega(d_R^3 d_L^3/L')/\log^{26}n$.
\end{claim}
\begin{proof}
Recall that $e(X_v,B'')\ge x_vd_L''/4$, which by assumption of the claim is at
least $2R$. Thus, \cref{fact:p2} implies that the number of $2$-paths in
$X_v\times B''\times X_v$ is at least
\begin{equation}
  \label{eq:case1P2}
\Omega((x_v d_L'')^2 / R).
\end{equation}
Given a two path $(w_1,z,w_2)\in X_v\times B''\times X_v$, there are at
least
$(\delta_v/2)(\delta_v/2-1)$ ways to choose distinct $r_1\in N_{B'}(w_1), r_2\in
N_{B'}(w_2)$. In this claim we are assuming $\delta_v\neq 1$ and consequentially
$\delta_v\ge 4$. Thus, ${(\delta_v/2)(\delta_v/2-1)\ge \delta_v^2/8}$.
Choosing $r_1,r_2$ in this manner, the vertices $(v,r_1,w_1,z,w_2,r_2)$ form a
hexagon containing a single vertex from $A'$. Combined with
\cref{eq:case1P2} this implies
\begin{equation}
  \label{eq:case1XX2}
\hex(v)\ge \Omega(\delta_v^2 (x_vd_L'')^2/R  ).
\end{equation}
Using \cref{eq:deltavxv} and \cref{eq:primes-bigger} in \cref{eq:case1XX2} gives
\[ \hex(v) \ge \frac{1}{\log^2 n}\Omega( (d_L' d_R' d_L'')^2 / R ) \ge
\frac{1}{\log^{26}n}\Omega(d_R^3 d_L^3/L'). \] 
\end{proof}

\begin{figure}[htpb]
  \centering
  \begin{subfigure}[b]{0.45\textwidth}
    \includegraphics[width=\textwidth]{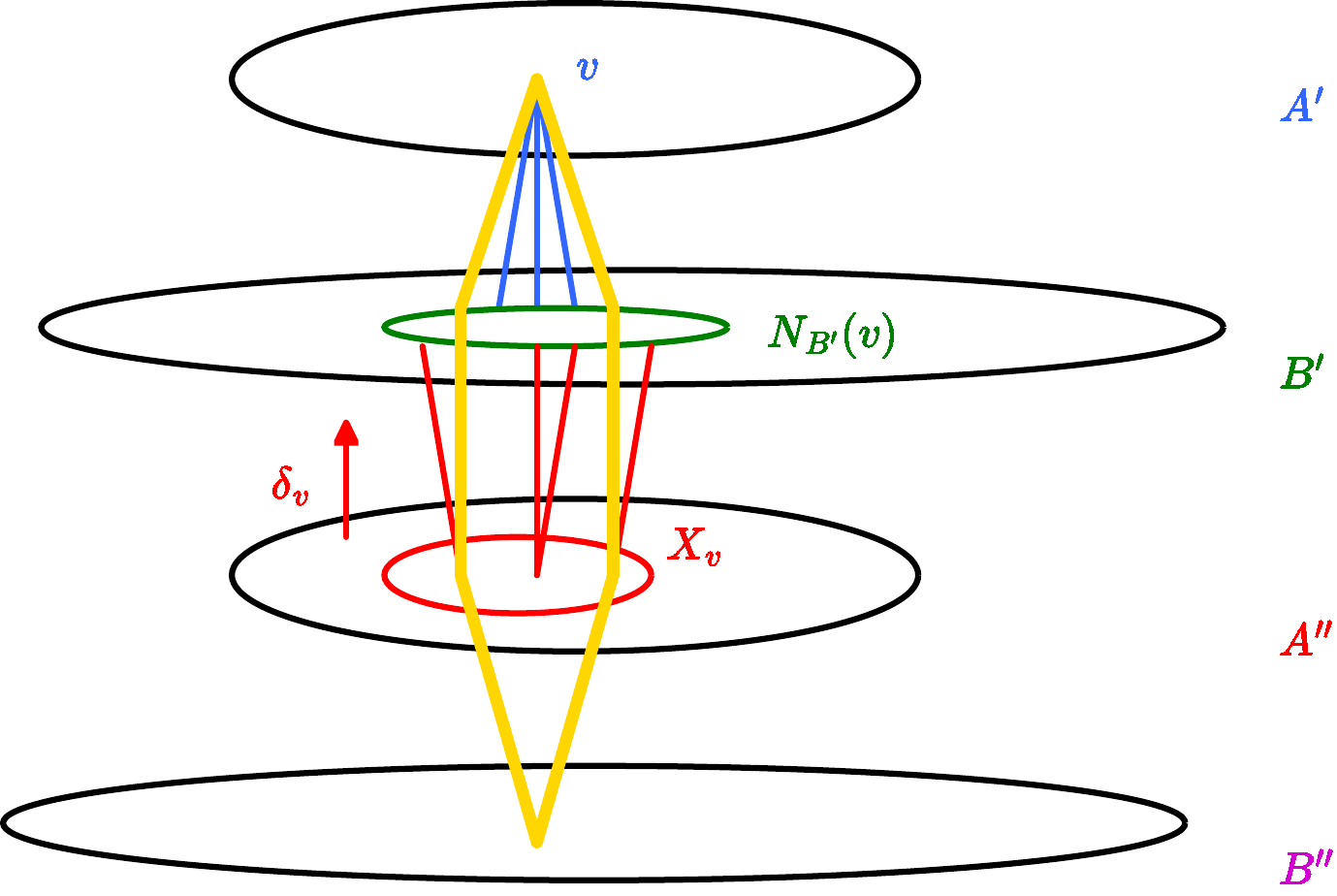}
    \caption{\cref{clm:case1}}
  \end{subfigure}
  \hfill
  \begin{subfigure}[b]{0.45\textwidth}
    \includegraphics[width=\textwidth]{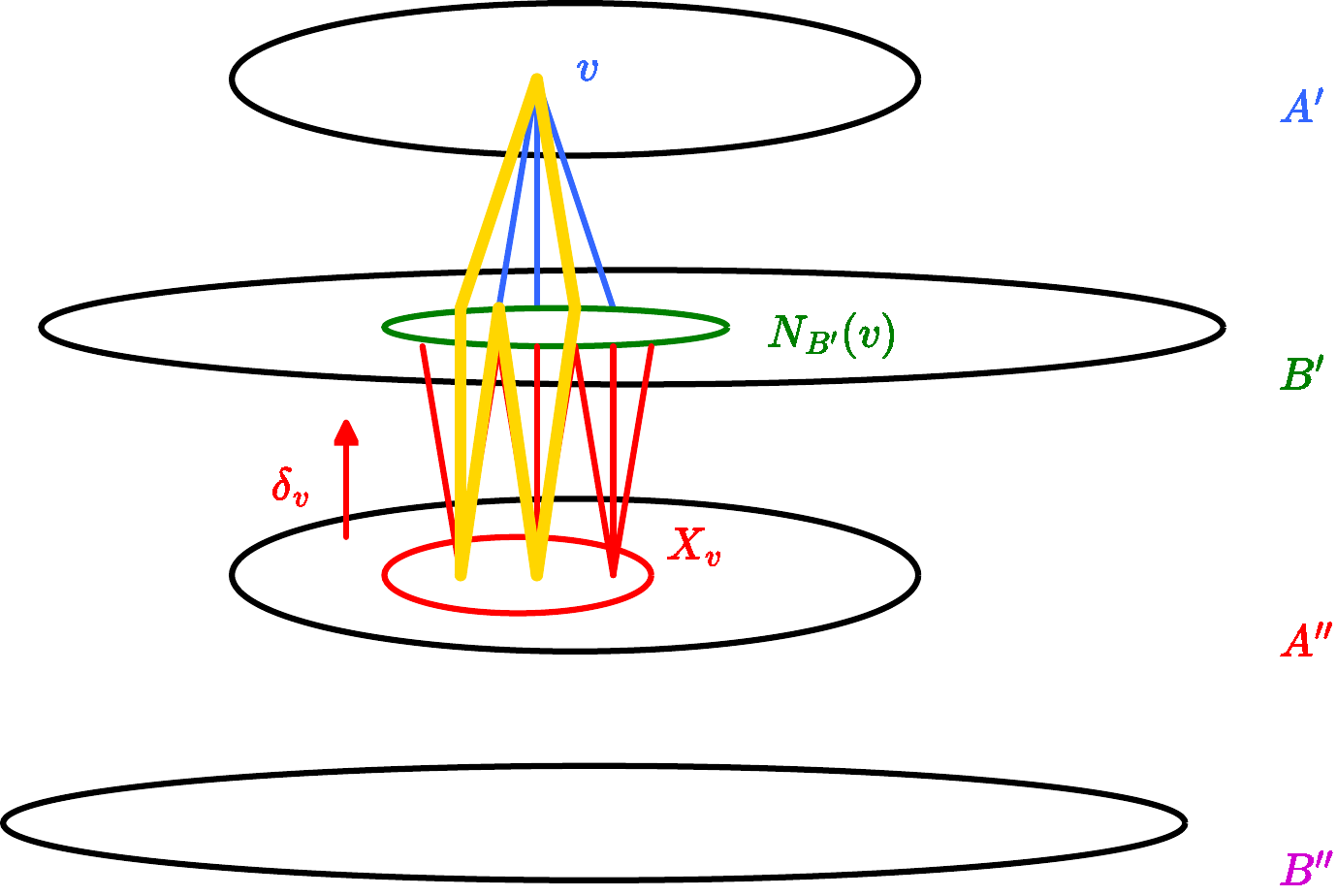}
    \caption{\cref{clm:case2}}
  \end{subfigure}
  \caption{Two methods for showing that $\hex(v)$ is large}
  \label{fig:case12}
\end{figure}

Before proceeding to the next case, we make some simple observations.
\begin{fact}\label{fact:dldrbig}
    $d_R, d_L \ge \log^{10} n \ge 10^{14}$.
\end{fact}
\begin{corollary}\label{cor:dldrdloverr}
$d_L' d_R' d_L''/R \ge 10^{15} \log n$.
\end{corollary}
\begin{proof}
By \cref{eq:primes-bigger} and \cref{fact:dldrbig} we have
\[d_L' d_L'' d_R' / R \ge (d_L^2 d_R/R)/\log^{12} n \ge (m^3/(LR)^2)/\log^{12} n \ge \log^{18} n \ge 10^{15}\log n.\]
\end{proof}
\begin{corollary}\label{cor:drprime_big}
$d_R'\ge 10^7 \log n$.
\end{corollary}
\begin{proof}
By \cref{eq:primes-bigger} and \cref{fact:dldrbig} we have
\[d_R' \ge d_R/\log^4 n \ge \log^6 n \ge 10^7 \log n.\]
\end{proof}

\begin{claim}\label{clm:case2}
Suppose vertex $v\in A'$ has $x_vd_L''/4 \le 2R$. Then
$\hex(v)\ge \Omega((d_Ld_R)^3 d_R^3/(LL'))/\log^{40}n$.
\end{claim}
\begin{proof}
We claim that if $x_v$ is this small then we have supersaturation of $4$-paths between
$N_{B'}(v)$ and $X_v$ that start in $N_{B'}(v)$. 
Indeed, \cref{eq:deltavxv}, \cref{cor:dldrdloverr}, \cref{cor:drprime_big} and our assumption that $x_v$ is small imply
\begin{equation*}
 e(N_{B'}(v), X_v) \ge \frac{1}{16\log n}d_L' d_R'\ge 10^{12}R/d_L'' + 10^5 d_L'\ge 50(d_L' + 8R/d_L'') \ge 50(d_L' + x_v).
\end{equation*}
Applying \cref{fact:p4}, \cref{eq:deltavxv}, the
hypothesis that $x_v$ is small, and \cref{eq:primes-bigger}, the number of such
$4$-paths is at least
\[ \frac{1}{\log^{4}n}\Omega\left( \frac{(d_L'd_R')^{4}}{d_L' x_v^2} \right)
\ge\Omega((d_Ld_R)^3 d_R^3/(LL'))/\log^{40}n. \] 
Each of these $4$-paths paired with $v$ gives a hexagon that contributes to
$\hex(v)$. This yields the desired bound on $\hex(v)$.
\end{proof}

The remaining two cases for what the BFS out of $v$ looks like will both have $\delta_v=1$. 
When $\delta_v=1$, we have the following interesting property:
\begin{claim}\label{clm:deltaisone}
Fix $v\in A'$ with $\delta_v=1$.
Then, $|B_v'|\ge d_L'/\log^{5}n$ and $\beta_v \ge d_R'/\log^{5}n$.
\end{claim}
\begin{proof}
We count $e(B_{v'}, X_v)$ in two ways. 
First, 
\begin{equation}
  \label{eq:d1--exp1}
e(B_{v'}, X_v) \le |B_{v'}|\beta_v.
\end{equation}
Second, combining \cref{eq:BvprimeXv} and \cref{eq:deltavxv} we have
\begin{equation}
  \label{eq:d1--exp2}
e(B_{v'}, X_v) \ge \frac{1}{\log^3 n} x_v \ge \frac{1}{16\log^{4}
n} d_L' d_R'.
\end{equation}
Combining \cref{eq:d1--exp1} and \cref{eq:d1--exp2} gives
\begin{equation}
  \label{eq:Bbeta}
|B_{v}'|\beta_v \ge \frac{1}{\log^{5} n} d_L' d_R'.
\end{equation}
Now we recall upper bounds for $|B_{v'}|,\beta_v$.
Applying $\beta_v \le d_R'$ in \cref{eq:Bbeta} gives $|B_{v}'| \ge d_L'/\log^{5} n.$
Applying $|B_{v}'|\le d_L'$ in \cref{eq:Bbeta} gives $\beta_v \ge d_R'/\log^{5}n.$


\end{proof}

\begin{figure}[htpb]
  \centering
  \begin{subfigure}[b]{0.45\textwidth}
    \includegraphics[width=\textwidth]{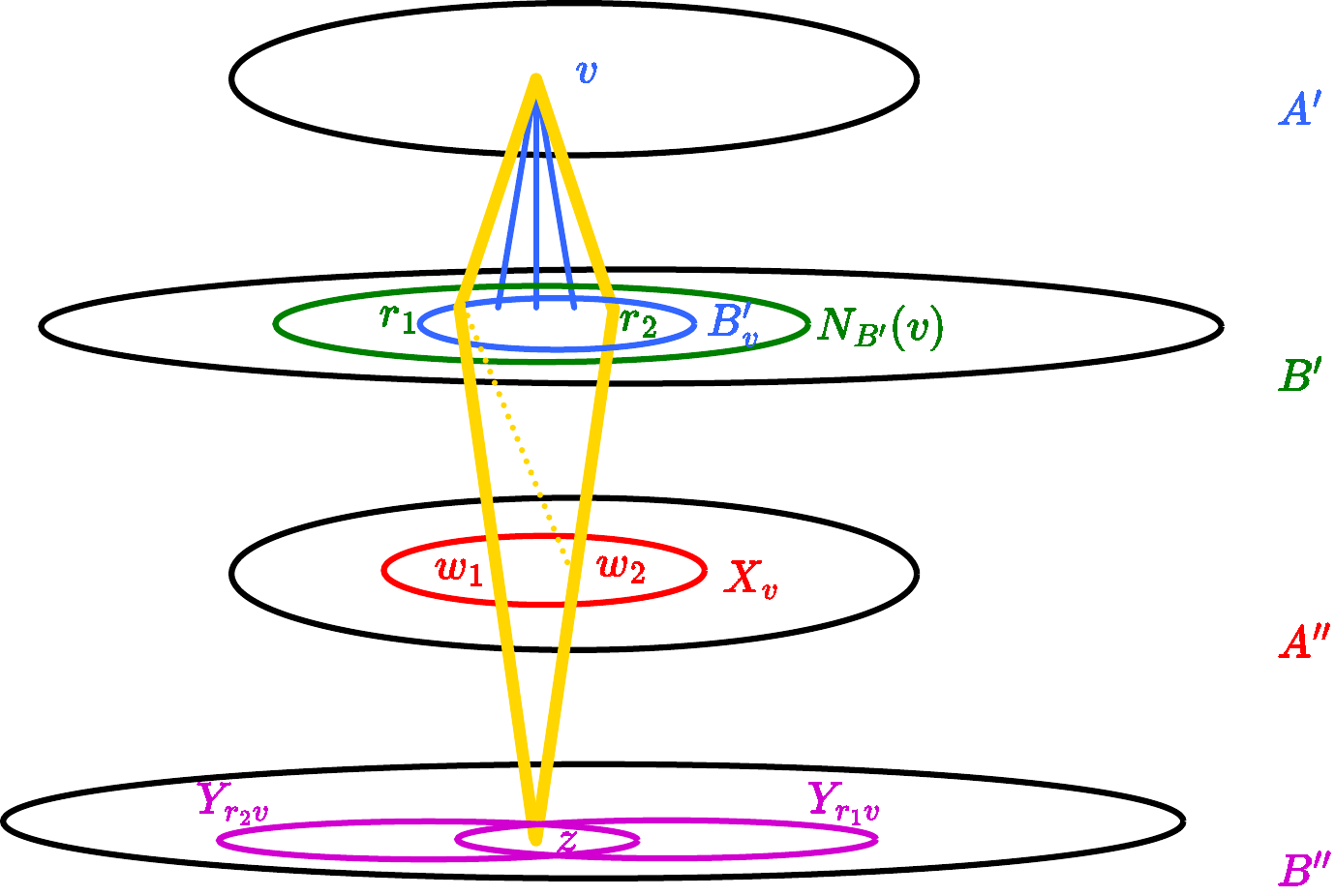}
    \caption{\cref{clm:case3}}
  \end{subfigure}
  \hfill
  \begin{subfigure}[b]{0.45\textwidth}
    \includegraphics[width=\textwidth]{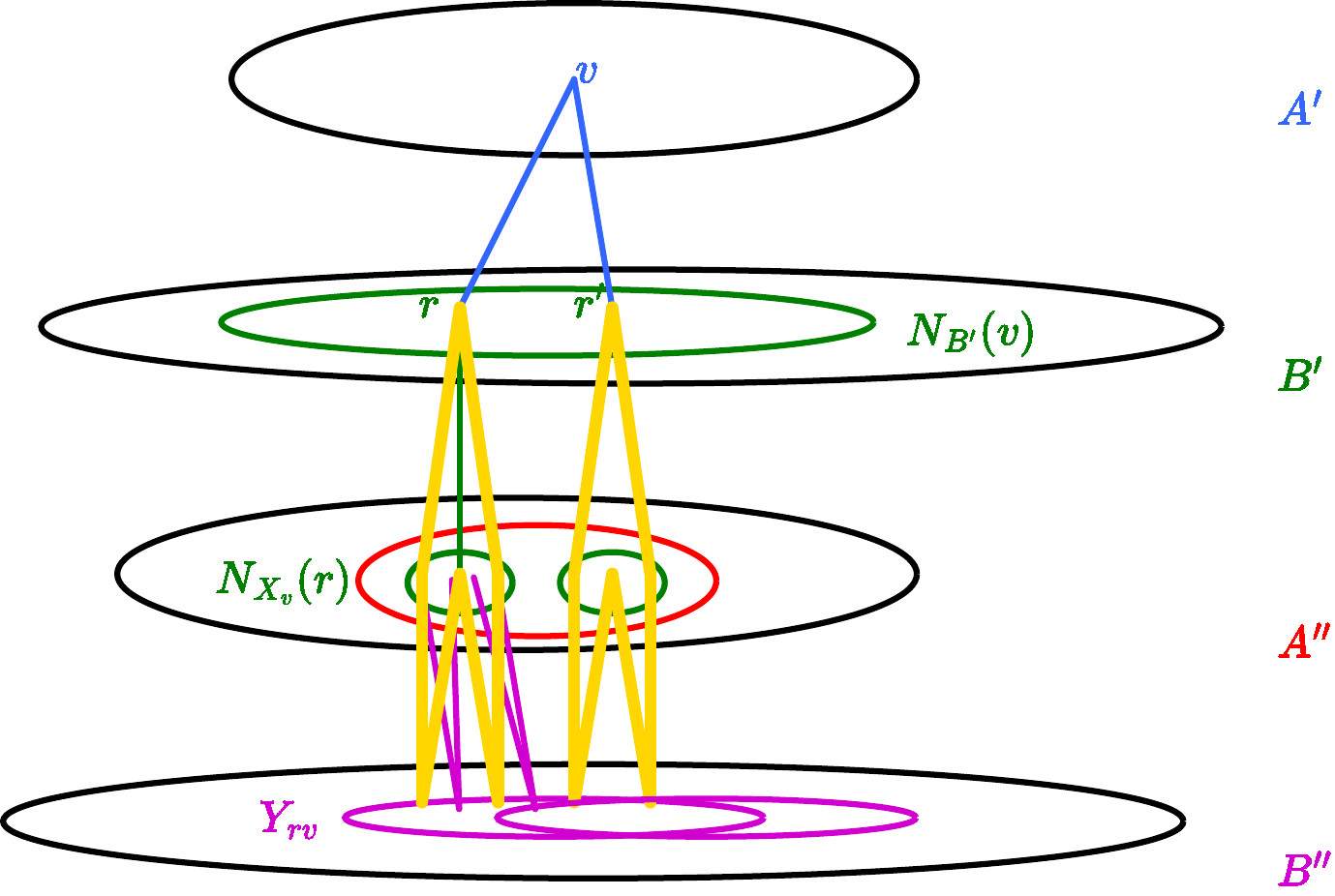}
    \caption{\cref{clm:case4}}
  \end{subfigure}
  \caption{Handling $\delta_v=1$}
  \label{fig:case34}
\end{figure}

We are now ready to handle the cases where $\delta_v=1$.
\begin{claim}\label{clm:case3}
Fix $v\in A'$ and suppose $y_v d_L' > (4\log^{5} n) R$ and $\delta_v=1$. 
Then $\hex(v)\ge \Omega((d_L^3 d_R^3)/L')/\log^{46}n$.
\end{claim}
\begin{proof}
We construct an auxiliary graph $H$ between $B_v'$ and $B''$ as follows: place an
edge between $r\in B_v'$ and $z\in B''$ if $r,z$ have a common neighbor in
$X_v$. We have, $|E(H)|\ge |B_v'|y_v/2$. 
By \cref{clm:deltaisone}, and the hypothesis of the claim we have 
\begin{equation}
  \label{eq:EHbigbig}
|E(H)| \ge \frac{1}{2\log^{5}n} d_L' y_v \ge 2R.
\end{equation}
Thus, applying \cref{fact:p2} to \cref{eq:EHbigbig}, the number of $2$-paths in
$H$ starting in $B_v'$, is at least
\begin{equation}
  \label{eq:case3aux}
\frac{1}{\log^{10}n} \Omega\left(  (d_L' y_v)^2 / R. \right).
\end{equation}
Now, fix some $2$-path $(r_1,z,r_2)\in B_v'\times B''\times B_v'$ in $H$.
There are at least $\ceil{\lambda_v/4}^2$ pairs $w_1,w_2\in X_v$ such that
$w_i\in N(z)\cap N(r_i)$ for $i=1,2$. 
We claim that for each such $w_1,w_2$, the tuple $(v,r_1,w_1,z,w_2,r_2)$ is a hexagon.
This is not immediately obvious: we need to rule out the possibility
of $w_1=w_2$ (in which case we would have a $4$-cycle with a dangling edge
rather than a true hexagon).
Fortunately, we can rule this out: the fact $\delta_v=1$ precisely means that
for $i=1,2$, $w_i$ has only $1$ neighbor in $B_v'$.
Thus, we get at least $\ceil{\lambda_v/4}^2$ times the quantity in
\cref{eq:case3aux} many hexagons. That is,
\begin{equation}
  \label{eq:chargecase3}
\hex(v)  \ge  \frac{1}{\log^{10} n} \Omega((d_L' y_v \lambda_v)^2/R).
\end{equation}
Now, recall that \cref{eq:lambdarvyrv} gives a bound on $y_v \lambda_v$: it
asserts that for any $r\in B'_v$ we have
\begin{equation}
  \label{eq:case3captureedges}
y_v \lambda_v \ge y_{rv} \lambda_{rv}
\ge \frac{1}{\log n}e(N_{X_v}(r), B'') 
\ge \frac{1}{\log n} |N_{X_v}(r)| d_L''/4
\ge \frac{1}{\log n} (\beta_v/2) d_L''/4.
\end{equation}
Now, using \cref{clm:deltaisone} we have
\begin{equation}
  \label{eq:case3simpler}
 y_v \lambda_v \ge \frac{1}{\log^{6} n}\Omega(d_R' d_L'').
\end{equation}
Now, we can bound \cref{eq:chargecase3} by
\[ \hex(v)\ge  \frac{1}{\log^{22} n}\Omega((d_L' d_R' d_L'')^2)/R.\]
Now, using \cref{eq:primes-bigger} we have
\begin{align*}
\hex(v) &\ge \frac{1}{\log^{22} n} \Omega((d_L' d_R' d_L'')^2)/R\\
&\ge \frac{1}{\log^{26}n}\Omega(d_L' (d_R'd_L'')^2 (m/L')/R)\\
&\ge \frac{1}{\log^{46}n}\Omega(d_L (d_Rd_L)^2 (m/L')/R)\\
&= \frac{1}{\log^{46}n}\Omega(d_L^3 d_R^3 /L').
\end{align*}


\end{proof}

The final case is slightly different. Instead of bounding $\hex(v)$ we will
directly bound $\hex(G)$.
\begin{claim}\label{clm:case4}
Fix $v\in A'$ with $\delta_v=1$, 
$y_vd_L' \le (4\log^{5} n)R$, and $x_v d_L''> 8R$.
Then $\hex(G) \ge \Omega((d_L d_R)^5 /L^2)/\log^{67}n$.
\end{claim}
\begin{proof}
We will find hexagons of the following form:
Take $r\in B_v'$, and then let  $(x_1,x_2,x_3,x_4,x_5)$ be a $4$-path between
$N_{X_v}(r)$ and $Y_{rv}$ (starting in $N_{X_v}(r)$). Then,
$(r,x_1,x_2,x_3,x_4,x_5)$ is a hexagon.
We claim that we have $4$-path supersaturation between $N_{X_v}(r)$ and
$Y_{rv}$. Indeed, by \cref{eq:lambdarvyrv} we have
\[ e(N_{X_v}(r), Y_{rv})\ge \frac{1}{\log n} e(N_{X_v}(r),
B'') \ge \frac{1}{4\log n}|N_{X_v}(r)|d_L''\ge \frac{1}{8\log n}\beta_v d_L''. \] 
Then, by \cref{clm:deltaisone} we have
\begin{equation}
  \label{eq:case4-edgect}
e(N_{X_v}(r), Y_{rv})\ge \frac{1}{8\log^{6} n} d_R' d_L''.
\end{equation}
Now, recall that $|N_{X_v}(r)|\le d_R'$, so \cref{eq:case4-edgect} together with
\cref{eq:primes-bigger} and \cref{fact:dldrbig} imply $e(N_{X_v}, Y_{rv})\ge 100|N_{X_v}(r)|$. 
Next, observe that by \cref{eq:primes-bigger}, \cref{eq:case4-edgect} implies
$$|Y_{rv}|\le y_v \le (4\log^{5}n)R/d_L' \le \frac{1}{4\log n}\frac{LR}{m}\le
\frac{1}{4\log^{31}n}\frac{m^2}{LR}\le
\frac{d_Ld_R}{\log^{31}n}\le\frac{d_R' d_L''}{800\log^{6} n} \le \frac{e(N_{X_v}(r),Y_{rv})}{100}.$$
Combined, we have
\[ e(N_{X_v}(r), Y_{rv}) \ge 50(|Y_{rv}|+|N_{X_v}(r)|), \] 
which is the required condition for $4$-path supersaturation.
Now, applying \cref{fact:p4} to \cref{eq:case4-edgect} the number of $4$-paths
that we get is:
\begin{equation}
  \label{eq:case4-countingp4s}
\Omega \left( \frac{e(N_{X_v}(r), Y_{rv})^{4}}{|N_{X_v}(r)||Y_{rv}|^2} \right) 
\ge \frac{1}{\log^{24} n} \Omega\left( \frac{(d_R' d_L'')^4}{\beta_v y_v^2} \right) 
\end{equation}
Applying our assumption on $y_v$ and the bound $\beta_v\le d_R'$ gives
\begin{equation}
  \label{eq:case4-almost}
\frac{1}{\log^{34} n} \Omega\left( \frac{(d_R' d_L'')^4}{d_R' (R/d_L')^2} \right) 
\end{equation}
Applying \cref{eq:primes-bigger} to \cref{eq:case4-almost}, the number of such
$4$-paths is at least 
\[ \frac{m^{9}}{L^{6}R^{5}} \Omega(1/\log^{61}n). \] 
Each $r\in B_{v}'$ contributes this many hexagons, and the contributed hexagons
are distinct. 
Applying \cref{clm:deltaisone} we have
\[ \hex(G)\ge |B_v'| \frac{m^{9}}{L^{6}R^{5}}\Omega(1/\log^{61}n) \ge
\frac{m^{10}}{L^{7}R^{5}}\Omega(1/\log^{67}n)\ge
(d_Rd_L)^3(d_L d_R/L)^2/\log^{67}n. \] 
\end{proof}




Now, we can easily combine \cref{clm:case1}, \cref{clm:case2}, \cref{clm:case3},
\cref{clm:case4} to deduce the theorem. Indeed, if any vertex $v\in A'$ falls
into the case described by \cref{clm:case4} then have the desired bound on
$\hex(G)$. Otherwise, each vertex must fall into one of the cases covered in
\cref{clm:case1}, \cref{clm:case2}, \cref{clm:case3}. Picking whichever case is
most common of these, we can sum over the $\hex(v)$ bounds for at least
$|A'|/3=L'/3$ many $v\in A'$. This gives the desired bound on $\hex(G)$.

\end{proof}

\section{Progress Towards Listing Algorithms}\label{sec:listing_progress}
In this section we describe how our partial progress towards
\cref{conj:supersat} in \cref{thm:partialsupersat} can be used together with our simplified analysis capped
$k$-walks analysis from \cref{sec:simple_capped} to give improved listing
algorithms for $C_6$'s in sparse graphs. Our main result is:

\listfast

\cref{thm:listfast} follows immediately from \cref{lem:listlist} with
$\Delta=m^{.4}$ and the following lemma:
\begin{lemma}
Fix $G$ with $\Delta(G)\le m^{.4}$. The number of capped $3$-walks in $G$ is at
most $\tilo(m^{1.6}+t)$.
\end{lemma}

\begin{proof}
If $n\le 1000$ the result is trivial; we assume this is not the case.
Fix $\lambda = 200\ceil{\log n}^{15}$.
As in \cref{thm:cappedwalks} and \cref{thm:morecappedwalks} we use
\cref{lem:DKSbuckets} to obtain 
$\mc{B}=(V_*,d_*, G'$, $X_1,X_2, X_3, d_1,d_2,d_3)$ with the properties
described in \cref{lem:DKSbuckets}. 
The number of $3$-walks in $X_1\times X_2 \times X_3  \times V(G')$ is at most
$|X_1| \cdot d_1d _2 d_3$; recall that the number of capped $3$-walks in $G$ is
at most $\polylog(n)|X_1|d_1d_2d_3$ --- thus, in the rest of this proof we'll
focus on bounding $|X_1|d_1d_2d_3$. At this point we break into 576 cases (with
the help of a computer). A \defn{case} is parameterized by a tuple
$(B_{12},B_{23},D_{12},D_{23}, U_{12},U_{23}, A_{12}, A_{23})$,
whose entries can take on the following values:
\begin{itemize}
\item $B_{12}\in \set{1,2}, B_{23}\in \set{2,3}$. $B_{12}$ indicates which of $X_1,X_2$ is bigger ($B$ stands for ``bigger''). More precisely, $|X_{B_{12}}| = \max(|X_1|, |X_2|)$. $B_{23}$ is defined analogously.
\item $D_{12}, D_{23}\in \set{1,2,3}$. $D_{12}, D_{23}$ indicate the ``density
regime'' ($D$ stands for ``density'') that the graphs between $(X_1,X_2)$ and
$(X_2,X_3)$ fall into. This means, letting $R=\max(|X_1|,|X_2|),
L=\min(|X_1|,|X_2|), m_{12}=e(X_1,X_2)$ and letting $v_1=1, v_2=(d_R^3/L)/\lambda^3,
v_3=(d_L^2d_R^2/L^2)/\lambda^4$, we have that $v_{D_{12}} = \min(v_1,v_2,v_3)$.
$D_{23}$ is defined analogously.
\item $U_{12}\in \set{0,1},U_{23}\in\set{0,1}$. $U_{12}$ indicates whether the
graph between $X_1,X_2$ is ``balanced'' or ``unbalanced'' ($U$ stands for
``unbalanced''). More precisely, $U_{12}$ is $1$ if $R\ge (RL)^{2/3}$ and $0$ if $R \le (RL)^{2/3}$.
$U_{23}$ is defined analogously.
\item $A_{12},A_{23}\in \set{0,1}$. $A_{12}$ indicates whether we can charge any
hexagons to the graph between $X_1, X_2$ ($A$ stands for ``any''). More precisely, $A_{12}$ is $1$ if $m\ge \lambda\max(|X_1|,|X_2|,(|X_1||X_2|)^{2/3})$ and $0$ if $m\le
\lambda\max(|X_1|,|X_2|,(|X_1||X_2|)^{2/3})$. $A_{23}$ is defined analogously.
\end{itemize}
Note that $\mc{B}$ falls into at least one case (and the cases are slightly
overlapping on the equality cases).
Now, let $\mc{C}=(B_{12},B_{23},D_{12},D_{23}, U_{12},U_{23},A_{12},A_{23})\in
\set{1,2}\times\set{2,3}\times\set{1,2,3}^2\times \set{0,1}^4$ be a case that describes $\mc{B}$.
For $i=1,2,3$, let $x_i = \log_{m}(|X_i|), \delta_i = \log_m(d_i/\lambda)$. Let
$\delta_* = \log_m(d_*/\lambda)$, and let $\tau$ denote
$\log_{m}(t/c)$ where $t$ is the number of $C_6$'s in $G$, and
$c$ is the constant from \cref{thm:partialsupersat}.
We now outline a series of \emph{linear} constraints that
$(x_1,x_2,x_3,\delta_1,\delta_2,\delta_3,\tau)$ must satisfy.

\begin{align*}
&x_1,x_2,x_3\in [0,1] & \forall i,|X_i|\le n\le m.\\
&\delta_*,\delta_1,\delta_2,\delta_3\in [0,.4] & \text{$\Delta(G)\le m^{.4}$.}\\
&\tau\ge 0\\
&\forall i\in[3], \delta_i \le \delta_*. &\text{$\Delta(G')\le\delta_*$}.\\
&\delta_*+x_1\le 1,&|X_1|\delta_*\le 2m\\
&\forall i\in [3], \delta_i + x_i \le 1& \forall i, |X_i|d_i\le 2m.
\end{align*}

Next, based on $B_{12},B_{23}$ we get ordering relationships between $x_1,x_2,x_3$.
For instance, if $B_{12}=1,B_{23}=3$ then we add the constraints
\[x_1\ge x_2, x_2\le x_3.\]
We get more constraints from $U_{12},U_{23}$. For instance, if
$U_{12}=0,U_{23}=1$ we add the following constraints:
\[ \max(x_1,x_2)\le 2\min(x_1,x_2), \max(x_2,x_3)\ge 2\min(x_2,x_3). \]
Note that while these constraints appear non-linear (due to the $\max,\min$)
they are actually linear, because $B_{12},B_{23}$ already specify orders between
$x_1,x_2$, and between $x_2,x_3$! So, based on the value of $B_{12},B_{23}$ we
could ``unroll'' these $\min/\max$'s to the appropriate $x_i$; we decline to do so here for clarity.

Next, we implement constraints based on $A_{12},A_{23}$. If $A_{12}=1$ the constraint will be
\[x_1+\delta_1 \ge \max(x_1, x_2, (2/3)(x_1+x_2)).\]
The constraint would be flipped if $A_{12}=0$. The constraint is analogous for $A_{23}$. 
Again, the $\max$ can be unrolled due to the value of the $\max$ being determined by $B_{12},U_{12}$.

Next we implement constraints based on $D_{12},D_{23}$.
For instance, if $D_{12}=1$ we would add constraints (where, as always, the $\max,\min$'s can be unrolled):
\[0\le 3(x_1+\delta_1)-3\max(x_1,x_2)-\min(x_1,x_2), 0\le 4(x_1+\delta_1)-2\max(x_1,x_2)-4\min(x_1,x_2).\]

\paragraph{Collecting Hexagons}
Now, with all these constraints in place, we will use \cref{thm:partialsupersat}
to obtain a lower bound on the number of hexagons in our graph.
Then we'll set up an optimization problem that tells us the slowest possible
time for our algorithm, observe that this optimization problem is an LP in each
case, and solve the LPs on a computer to find a bound on the performance of our algorithm.

If $A_{12}=1$, then letting $\ell = \min(x_1,x_2), r=\max(x_1,x_2), e=x_1+\delta_1$ we add the constraint
\[\tau \ge 6e-3(r+\ell)+ \min(0, 3e-3r-\ell, 4e-4\ell-2r).\]
This constraint is implied by \cref{thm:partialsupersat}, and we satisfy the
pre-condition of \cref{thm:partialsupersat} because $A_{12}=1$.
Note that as usual, the $\min/\max$'s can be unrolled based on the case. 
If $A_{23}=1$ we add an analogous constraint on $\tau$.


\paragraph{Bounding capped-$3$ walks}.
Now, we want to answer the following question:
``subject to $(x_1,x_2,x_3,\delta_*,\delta_1,\delta_2,\delta_3)$ satisfying the
above constraints, is it possible that the number of $3$-walks in $X_1\times
X_2\times X_3\times V(G')$ is larger than $\tilo(m^{1.6}+t)$?
We can answer this question in the negative by requiring the number of $3$-walks
to be more than $\Omega(t\log^{45}n)$, and showing that subject to this
constraint, the number of $3$-walks cannot be larger than $O(m^{1.6})$.
In our linear programming formulation, this corresponds to adding the constraint
\begin{equation*}
x_1+\delta_1+\delta_2+\delta_3 \ge \tau,
\end{equation*}
and optimizing the following objective function, which is the number of
$3$-walks (up to $\poly\log n$ factors):
\[x_1+\delta_1+\delta_2+\delta_3.\]
If the objective function optimizes to a number which is at most $1.6$, then it
will prove that the number of $3$-walks of the specified form cannot be larger
than $\tilo(m^{1.6}+t)$.

So, we can prove the lemma by iterating over the $576$ possible cases, defining
the appropriate linear program in each case, and checking that the optimum of
the linear program is at most $1.6$.
We wrote a simple program \cite{code} to iterate over
the cases and solve the linear programs.
Running the code we find that the optimum of each linear program is at most
$1.6$, as desired. We leave developing a human interpretable proof as an open question.
\end{proof}

\bibliography{refs}
\appendix
\section{Listing $C_{2k}$'s Using the Supersaturation Conjecture}\label{sec:listing_with_supersat}
Jin and Zhou \cite{supersat_observation} observed that the unbalanced
supersaturation conjecture would allow one to obtain conditionally optimal
sparse / dense listing algorithms for $C_{2k}$'s.
In this section we give a simple proof of this observation, using the techniques
of \cref{sec:simple_capped}.
Recall the unbalanced supersaturation conjecture:

\supersatconj
Conditional on \cref{conj:supersat} we can follow the same proof strategy from
\cref{sec:simple_capped} to obtain a $C_{2k}$ listing algorithm.
First, we need a simple consequence of \cref{conj:supersat}.
\begin{corollary}\label{cor:supersatnondisjoint}
Assume \cref{conj:supersat}.
Fix $G$ and $A,B\subseteq V(G)$ \emph{not necessarily
disjoint}.
Then,
\[ m \le O(|A|+|B|+(|A||B|)^{(k+1)/(2k)} + t^{1/(2k)}\sqrt{|A||B|}).\] 
\end{corollary}
\begin{proof}
As in \cref{cor:non-disjoint-extreme}, we pass to a bipartite subgraph with at least $m/2$ edges, while only decreasing $|A|,|B|,t$.
Then we bound $m/2$ using \cref{conj:supersat}.
\end{proof}
Next we prove an analogue of \cref{lem:keylem}, for analyzing the structure
between two layers in the decomposition of \cref{lem:DKSbuckets}.
\begin{lemma}\label{lem:keylem2}
Assume \cref{conj:supersat}
Let graph $G$ have $\Delta(G)\le m^{2/(k+1)}$.
Fix $A\subseteq V(G)$ and $d\le \Delta(G)$.
Let $B$ denote the set of vertices $v\in V(G)$ that have $|N_A(v)|\in
[d/2,d]$. Then, 
\[ d\sqrt{|B|}/\sqrt{|A|} \le \O(m^{\frac{1}{k+1}} + t^{1/(2k)}). \] 
\end{lemma}
\begin{proof}
The proof is nearly identical to that of \cref{lem:keylem}.
The only difference is that instead of applying the bound
\[e(A, B)\le O(|A| + |B| + (|A||B|)^{(k+1)/(2k)}),\]
which applied when we assumed the graph was $C_{2k}$-free, we apply
\cref{cor:supersatnondisjoint}, which gives:
\[ e(A, B) \le O(|A|+|B|+(|A||B|)^{(k+1)/(2k)} + t^{1/(2k)}\sqrt{|A||B|}).\] 
Substituting this new bound on $e(A,B)$ into the proof of \cref{lem:keylem} and
otherwise leaving the analysis unchanged gives the desired bound on
$d\sqrt{|B|/|A|}.$
\end{proof}
Now we deduce the analog of \cref{thm:cappedwalks}.
\begin{theorem}\label{thm:morecappedwalks}
Assume \cref{conj:supersat}.
Let $G$ be a graph with $\Delta(G)\le m^{2/(k+1)}$. Then, the number of capped $k$-walks in $G$ is at most $\tilo( m^{2k/(k+1)} + t ).$
\end{theorem}
\begin{proof}
Just like in the proof of \cref{thm:cappedwalks}
we apply \cref{lem:DKSbuckets} to obtain $V_*,d_*, G'$, $X_1,\ldots,
X_k\subseteq V(G'), d_1,\ldots, d_k$ with the properties described in
\cref{lem:DKSbuckets}. The number of $k$-walks in $X_1\times \cdots \times
X_k\times V(G')$ is at most $|X_1| \cdot \prod_{j=1}^k d_j.$
In \cref{thm:cappedwalks} we bounded this product by repeatedly applying
\cref{lem:keylem} (and then using an additional observation to bound $\sqrt{|X_1||X_k|}d_k$).
Using \cref{lem:keylem2} in place of \cref{lem:keylem} we have:
\begin{equation}\label{eq:holders}
|X_1|\prod_{j=1}^k d_j \le O( (m^{1/(k+1)}+t^{1/(2k)})^{k-1} \cdot m ).
\end{equation}
Applying H\"older's Inequality to \cref{eq:holders} proves the desired bound on capped $k$-walks.
\end{proof}

Finally, we have the analog of \cref{cor:detection}.
\begin{corollary}\label{thm:conditional_listing}
Assume \cref{conj:supersat}.
Then, there is an algorithm for listing $C_{2k}$'s with running time
$\tilo(t+\min(m^{2k/(k+1)}, n^2))$.
\end{corollary}
\begin{proof}
First, we note that it is sufficient to establish the existence of an
$\tilo(t+m^{2k/(k+1)})$ time algorithm. This is because balanced supersaturation
results (\cref{fact:vanillasupersat}) imply that if $m^{2k/(k+1)}>\Omega(n^2)$,
meaning $m>\Omega(n^{1+1/k})$, then $t>\Omega(m^{2k/(k+1)})$, and so the $t$
term dominates the running time. Now we give a listing algorithm with running time $\tilo(t+m^{2k/(k+1)})$.

Set $\Delta = m^{2/(k+1)}$. Assuming \cref{conj:supersat},
\cref{thm:morecappedwalks} implies that there are at most $\tilo( m^{2k/(k+1)} +
t )$ capped-$k$ walks starting at a vertex of degree at most $\Delta$.
Then, applying \cref{lem:listlist} gives a listing algorithm with the desired
running time.

\end{proof}

\section{Folklore Results about Path Supersaturation}\label{appendix:path-facts}

\begin{fact}
Let $G$ be a bipartite graph with parts $A,B$ of sizes $L,R$. If $m\ge 2R$, then
the number of $2$-paths in $A\times B\times A$ is at least $L(m/L)(m/R)/2$.
\end{fact}
\begin{proof}
Using Cauchy-Schwarz, the number of 2-paths is at least 
\[\sum_{v\in B} \deg(v)(\deg(v)-1) \ge \frac{1}{R} \left(\sum_{v\in B} \deg(v)\right)^2 - \sum_{v\in B} \deg(v) \ge m (m/R - 1) \ge m^2/(2R).\]
\end{proof}

\begin{fact}
Let $G$ be a bipartite graph with parts $A,B$ of sizes $L,R$. If $m\ge 50(R+L)$, then
the number of $4$-paths in $A\times B\times A\times B\times A$ is at least
$\Omega(L(m/L)^2(m/R)^2)$.
\end{fact}
\begin{proof}
Throughout the proof we use the notation $\deg_S(v) = |N(v)\cap S|$.
\begin{itemize}
\item Let $B_1$ be the set of vertices $v\in B$ with $\deg_{A}(v)\ge e(A,B)/(2|B|)$.
\item Let $A_1$ be the set of vertices $v\in A$ with $\deg_{B_1}(v)\ge e(A,B_1)/(2|A|)$.
\item Let $B_2$ be the set of vertices $v\in B_1$ with $\deg_{A_1}(v)\ge e(A_1,B_1)/(2|B_1|)$.
\end{itemize}
Observe that $e(A,B_1)\ge e(A,B)/2, e(A_1,B_1)\ge e(A,B)/4, e(A_1,B_2)\ge e(A,B)/8$.
We now demonstrate that there are a large number of $4$-paths starting in $A$. We can generate such a $4$-path as follows: 
\begin{itemize}
\item Fix an edge $(v_5, v_4)\in A_1\times B_2$.
\item Take $v_3\in N_{A_1}(v_4)$, $v_2\in N_{B_1}(v_3), v_1\in N_A(v_2)$.
\end{itemize}
The number of such walks is at least:
\[ \frac{e(A,B)}{8}\left(\frac{e(A,B)}{8|B_1|}-1\right) \left(\frac{e(A,B)}{4|A|}-1\right) \left(\frac{e(A,B)}{2|B|}-2\right) \ge\Omega(m^4/(LR^2)). \]

\end{proof}

\end{document}